%% file: nstpmpnp_main.tex
\documentclass[%
 aip,
 jmp,%
 amsmath,amssymb,
preprint,%
]{revtex4-1}

\usepackage{color}
\usepackage{graphicx}
\usepackage{tikz}
\usepackage{hyperref}
\usepackage{dcolumn}
\usepackage{bm}

\input{user_comm.tex}

\begin{document}

\preprint{AIP/123-QED}

\title[New porous medium Poisson-Nernst-Planck equations]{New porous medium Poisson-Nernst-Planck equations\\
for strongly oscillating electric potentials}

\author{M. Schmuck}
\affiliation{ 
Departments of Chemical Engineering and Mathematics,	Imperial College, London SW7 2AZ, UK,
}%
\email{m.schmuck@imperial.ac.uk.}


\date{\today}

\begin{abstract}
We consider the Poisson-Nernst-Planck system which is well-accepted for describing dilute electrolytes as well as 
transport of charged species in homogeneous environments. Here, we study these equations in porous media whose 
electric permittivities show a contrast compared to the electric permittivity of the electrolyte phase. Our main result is the 
derivation of convenient low-dimensional equations, that is, of effective macroscopic porous media 
Poisson-Nernst-Planck equations, which reliably describe ionic transport. 
The contrast in the electric permittivities between liquid and solid phase and the heterogeneity of the 
porous medium induce strongly oscillating electric potentials (fields). In order to account for this special 
physical scenario, we introduce a modified asymptotic multiple-scale expansion which takes advantage of the nonlinearly 
coupled structure of the ionic transport equations. This allows for a systematic upscaling resulting in a new effective porous medium formulation which shows a new transport term on the macroscale. Solvability of all arising equations is rigorously verified. This emergence of a new transport term indicates promising physical 
insights into the influence of the microscale material properties on the macroscale. Hence, systematic upscaling strategies
provide a source and a prospective tool to capitalize intrinsic scale effects for scientific, engineering, and 
industrial applications.

%
\end{abstract}

\maketitle


\section{Introduction}\label{sec:Intr}
\input{intro.tex}

\section{Proof of Theorem \ref{thm:AsExp}}\label{sec:Proof}
\input{proof.tex}


\section{Proof of Lemma \ref{lem:ExReMo}: Well-posedness of the upscaled model}\label{sec:WPM}
\input{existence.tex}

          \nocite{*}
          \bibliography{extpmPNP}
          %

\medskip

{\bf Acknowledgement.}\\
Main parts of this research were conducted by the author at the Chemical Engineering Department at MIT and
supported by the Swiss National Science Foundation (SNSF) under the prospective 
researcher Grant PBSKP2-12459/1.



\end{document}

%% file: user_comm.tex
\definecolor{grau}{gray}{.5}
\definecolor{schwarz}{gray}{0}

\newtheorem{thm}{Theorem}[section]

\newtheorem{lem}[thm]{Lemma}

\newtheorem{defn}[thm]{Definition}

\newtheorem{proof}{\emph{Proof.}}

\newtheorem{rem}[thm]{Remark}

\newcommand{\reff}[1]{(\ref{#1})}
\newcommand{\ol}[1]{\overline{#1}}

\newcommand{\cg}[1]{\mathcal{#1}}

\newcommand{\av}[1]{\left|#1\right|}

\newcommand{\Ll}[1]{\left\|#1\right\|}
\newcommand{\N}[2]{\left\|#1\right\|_{#2}}

\newcommand{\brkts}[1]{\left(#1\right)}
\newcommand{\ebrkts}[1]{\left[#1\right]}

\newcommand{\brcs}[1]{\left\{#1\right\}}

\newcommand{\abrkts}[1]{\langle #1\rangle}

\newcommand{\bsplitl}[2]{
\begin{equation}
\begin{split}
#1
\end{split}
\label{#2}
\end{equation}}
\newcommand{\bsplit}[1]{
\begin{equation*}
\begin{split}
#1
\end{split}
\end{equation*}}

%% file: intro.tex
The Poisson-Nernst-Planck (PNP) equations can be applied in many different 
physical contexts such as modeling of ionic transport, e.g. batteries, supercapacitors, fuel cells,
and capacitive desalination devices. Especially the fields of electrokinetics and electrohydrodynamics 
gained an increasing interest in recent years. Actual research aims to take advantage of 
scale effects in micro- and nano-fluidic devices for industrial applications and for the creation of 
chip-like devices (``lab on a chip''). 
Such devices can perform separation, mixing, and chemical analysis tasks. 
It is also possible to design electrokinetic 
pumps.\cite{Ajdari2000}

The study of geometric effects on the scale of cell membranes, muscles, and neurons 
by means of PNP equations finds currently a lot of 
attention in biology and medcine.\cite{Eisenberg1999,Gardner2000} The essential goal 
is to better understand how 
calcium ions, i.e. $Ca^{\,2+}$-ions, move in voltage-dependent calcium channels for example. 
These channels are a group of voltage-gated ion channels which can be found in muscles, 
glial cells, and neurons. Recent research attempts to mimic such biological ion channels with synthetically built channels.\cite{Siwy2003} For example by modifying channel geometry and surface charge one 
tries to better understand the effect of rectification. Rectificiation can be descriptively 
explained by the comparison of ionic flux with an electric current through 
a pn-diode. One usually studies rectfication factors (ratio of forward current to reverse 
current) in this context, see \onlinecite{Cheng2009} for example.

This broad range of applications in heterogeneous environments strongly rely on models 
which reliably and systematically account for effects of the microscale on the macroscale. A very common 
approach for deriving effective macroscopic equations is volume averaging.\cite{Pivonka2009,Quintard1994,Whitaker1986} Unfortunately, it is still unclear how to systematically treat nonlinear terms by this intuitive method. A technically slightly more involved 
approach is the homogenization method\cite{Bensoussans1978,Cioranescu2000,Pavliotis2008} which provides a reliable and systematic alternative under the assumption of a periodic pore distribution.

The general importance and the strong demand of properly upscaled equations in engineering, and 
design as well as optimization of scientific and industrial devices call for mathematical tools that rely on well-established principles for multiscale problems. Here, we want to systematically extend the widely accepted PNP system from the free space case towards solid-electrolyte composites showing a high contrast between its electric permittivities. For this purpose, we consider the full PNP equations with the help of a modified asymptotic two-scale expansion. This new approach accounts for the nonlinear and 
coupled structure of the system, see Theorem \ref{thm:AsExp}. As our main result, we derive the following new effective macroscopic porous 
medium PNP equations, that is,
\bsplitl{
\begin{cases}
p\partial_t {\rm u}^r_0
	-p\Delta {\rm u}^r_0 
	+{\rm div}\brkts{\mathbb{D}^r(t,x) \nabla{\rm u}^3_0}
	-{\rm div}\brkts{z_r{\rm u}_0^r\mathbb{M}\nabla{\rm u}^3_0}
	=0
	&\textrm{in }\Omega_T\,,
\\
-{\rm div}\brkts{\pmb{\epsilon}^0\nabla {\rm u}_0^3}
	=p\brkts{ {\rm u}_0^1-{\rm u}_0^2}
	&\textrm{in }\Omega_T\,,
\end{cases}
}{MaReIn}
which is valid under local (pore level) thermodynamic equilibrium and for arbitrary Debye lengths $\lambda_D>0$. The parameter $p$ denotes the porosity and $\mathbb{D}^r$, $\mathbb{M}$, and $\pmb{\epsilon}^0$ are effective transport 
tensors defined by the upscaling subsequently performed. The 
variable ${\rm u}_0^r$ represents effective macroscopic quantities such as the concentration of positively charged ions $n^+$ for $r=1$, the density of negatively charged ions $n^-$ for $r=2$, and the electrostatic potential $\Phi$ for $r=3$.

All equations appearing during the upscaling are rigorously justified by well-posedness criteria. In particular, Lemma \ref{lem:ExReMo} (Section \ref{sec:pmPNP}) 
guarantees the solvability of the new system \reff{MaReIn} which shows a new term $\mathbb{D}^r$ that accounts 
for a dominant influence of the oscillating electric potential on the concentrations.
We emphasize that this new term emerges as a result 
of an adapted asymptotic multi-scale expansion introduced in order to account for the heterogeneity induced 
by the porous medium and by the electric permittivities.
That means, in the classical asymptotic expansion
\bsplitl{
u_s(x) = u_0(x,x/s)
	+s u_1(x,x/s)
	+s^2 u_2(x,x/s)
	+\dots\,,
}{StAsEx}
we assume a special separation between the micro- and macro-scale in the terms $u_1$ and $u_2$ 
in difference to related literature\cite{Allaire2010,Looker2006,Ray2011,Schmuck2010} and classical 
homogenization theory,\cite{Cioranescu2000,Sanchez-Palencia1986} see \reff{AsExp} in Theorem \ref{thm:AsExp} below. 
We point out that the nonlinear character of the 
Nernst-Planck equations leads with \reff{StAsEx} to ill-posed reference cell problems without 
additional assumptions. Hence, we suppose that the reference cells, which define the 
micro-geometry of the porous medium, are in thermodynamic equilibrium. This guarantees then 
well-posedness. This solvability issue might also be the reason why 
the upscaling of the PNP equations was mainly restricted to thin double layer type 
approximations\cite{Holmes1990,Moyne2002} or linearized 
formulations in the context of Onsager reciprocal relations\cite{Onsager1931a,Onsager1931b} so far.\cite{Allaire2010,Looker2006} 
Physically, such situations occur when the electric 
permittivities between the electrolyte and the solid material are far apart, see Section \ref{sec:DiPes}.

\medskip

The article is organized in the following way: The dominant oscillating behavior of the electric potential 
is motivated by the contrast in the electric permittivities between the solid and the liquid phase in 
Section \ref{sec:DiPes} where also a related effective media theory is discussed. A historical overview of closely related upscaling results is given in 
Section \ref{sec:HiOvs}. In Section \ref{sec:Not}, we state elementary 
results and introduce necessary notation. The main results follow in Section \ref{sec:pmPNP}.  
Finally, we prove all results in Section \ref{sec:Proof}.

\subsection{Physical motivation: Dielectric permittivities of solids and liquids}\label{sec:DiPes}
\input{physics.tex}

\subsection{Review of related upscaling results of the PNP equations}\label{sec:HiOvs}

We briefly give a shortened historical overview, see Table \ref{tab:Cprsn}, and point out differences of contributions 
mainly based on the homogenization method. Table \ref{tab:Cprsn} examplifies how different upscaled formulations arise upon different assumptions initially (i.e., on the macroscale) made.

\begin{table}[htdp]
\caption{Effective PNP equations: Various different upscaled formulations arise under different assumptions. We abbreviate by ''PB'' the ''Poisson-Boltzmann equations'', by ''PNP'' the ''PNP equations'', and by ''Onsager'' the 
''Onsager reciprocal relations''.}
\tiny
\begin{center}
\begin{ruledtabular}
\begin{tabular}{ p{2.0cm}  p{2.8cm}  p{2.7cm}  p{3.2cm}  p{1.8cm} }
{\bf Article} 
	& {\bf PB} 
	& {\bf linearized PNP} 
	& {\bf PNP} 
	& {\bf Onsager}
\\\hline
Jackson et al. 1963, \cite{Jackson1963} 
	& \textbullet\, given periodic electric potential or 
	PB equations
	\newline
	\textbullet\, stationary case
	\newline
	\textbullet\, effective diffusion equation, no homogenization
	& 
	&  
	& 
\\\hline
Gross et al. 1968, \cite{Gross1968} 
	& \textbullet\, no homogenization 
	& 
	&
	& \begin{center} 
	\checkmark 
	\end{center}
\\\hline
M. H. Holmes 1990, \cite{Holmes1990} 
	& 
	&  
	& \textbullet\, generalization of Jackson et al. 1963 including surface charge
	\newline 
	\textbullet\, homogenization in Slotboom variables
	\newline
	\textbullet\, surface charge
	& 
\\\hline
Moyne et al. 2002, \cite{Moyne2002}, and Moyne et al. 2004, \cite{Moyne2004} 
	& 
	 & 
	 & \textbullet\, similar to Ray et al. 2011: Kind of thin double layer limit (PB in  reference cells)
	 	\newline 
	\textbullet\, surface charge
	&
\\\hline
Looker et al. 2006, \cite{Looker2006}
	&
	& \textbullet\, local equilibrium \newline
		\textbullet\, local electroneutrality \newline
		\textbullet\, fluid flow  \newline
		\textbullet\, surface charge
	&
	& \begin{center} \checkmark \end{center}
\\\hline	
Allaire et al. 2010, \cite{Allaire2010},
		based on Looker et al. 2006, \cite{Looker2006}
	&
	& \textbullet\, local equilibrium \newline
		\textbullet\, bulk electroneutrality \newline
		\textbullet\, fluid flow  \newline
		\textbullet\, surface charge
	&
	& \textbullet\, rigorous with two-scale convergence
\\\hline
Schmuck 2011, \cite{Schmuck2010}\newline
Schmuck 2012, \cite{Schmuck2012zamm}
	&
	&
	& \textbullet\, local equilibrium \newline
		\textbullet\, without surface charge
	&	
\\\hline
Ray et al. 2012, \cite{Ray2012} 
	&
	&
	& \textbullet\, Debye length as homogenization parameter
		 \newline
		 \textbullet\, without surface charge
		 \newline
		 \textbullet\, fluid flow
	&
\end{tabular}
\end{ruledtabular}
\end{center}
\label{tab:Cprsn}
\end{table}%

\medskip

Based on the stationary Nernst-Planck equations and a given periodic electric potential, 
the authors derive an effective diffusion equation for ion densities in \onlinecite{Jackson1963}. 
In \onlinecite{Gross1968}, a model for charged membranes separating dilute aqueous salt 
solutions is studied. The authors give expressions for Onsager's reciprocal relations\cite{Onsager1931a,Onsager1931b} without assuming 
small Debye-lengths nor a Debye-H\"uckel linearizaion of the Poisson-Boltzmann equation. 
Based on ideas from \onlinecite{Jackson1963}, Holmes\cite{Holmes1990} performs an asymptotic two-scale 
expansion for a PNP formulation rewritten in Slotboom 
variables. This work gives a very interesting approach for nonlinear diffusion in charged polymers. 
The results seem to be closely related to upscaling of electrostatic potentials $\phi_0(y)$ only
depending on the micro-scale $y$ (which is related to the use of periodic potentials 
in \onlinecite{Jackson1963}).

In \onlinecite{Moyne2002}, a macroscopic electrokinetic formulation describing electro-osmotic flow 
and electrophoretic motion in periodic porous media is obtained by the classical multiple-scale 
expansion method. They also perform a kind of thin double layer approximation in the reference cells by 
Poisson-Boltzmann equations. This approximation is well-known and frequently 
applied in electro-chemistry.\cite{Bazant2004a} The same authors\cite{Moyne2004} 
apply the periodic homogenization theory to upscale the 
Nernst-Planck equations in the medical and biological context of cartilage.

In \onlinecite{Looker2006}, the fundamental Onsager reciprocal relations\cite{Onsager1931a,Onsager1931b} together with positive definiteness of 
corresponding upscaled tensors are derived under the assumption of local thermodynamic equilibrium. 
Their starting point is a linearized \cite{O'Brien1978} $N$-component electrolyte in a dilute Newtonian 
fluid (small Reynolds number) flowing through a periodic porous medium. The authors also consider 
the physically interesting case of surface charge. However, the influence of the surface charge is not 
obvious at the end. The impact of such charges on the macroscopic level is of major scientific and engineering 
interest in microfluidics\cite{Mani2011,Dydek2011,Schmuck2012} and neurobiology\cite{Siwy2003}. 

Very recently, Allaire et al.\cite{Allaire2010} put the physically relevant derivations from \onlinecite{Looker2006} into the rigorous 
framework of the two-scale convergence. The main 
purpose is again the verification of Onsager's reciprocal relations as in \onlinecite{Looker2006}. 
Derivations of such relations require the assumption of local thermodynamic 
equilibrium, a linearized PNP system,\cite{O'Brien1978} and an electroneutrality 
assumption in the bulk which is physically closely related to a thin double layer approximation.

In \onlinecite{Ray2011}, the authors perform a singular limit with respect to the dimensionless 
Debye length. A weighted Debye length, i.e., 
$\lambda^\alpha$, and the use of $\lambda$ as the homogenization parameter has the meaning 
of upscaling the PNP system parallel to taking special ($\alpha>0$ arbitrary)  
thin double layer 
limits of the system. Espeically for $\alpha=1$, this is an interesting problem since the thin double layer 
approximation is a widely used simplification as already mentioned above.

\section{Notation and preliminaries}\label{sec:Not}
\input{notation.tex}

\subsection{Review of the classical PNP equations}\label{sec:ClPNP}
Before we come to the main results in this article, we briefly recall basics about the PNP system.
In view of computational convenience (block matrix solvers, e.g. \onlinecite{Schmidt2005}), notational clearity 
and compactness, and a non-linear (i.e., non-symmetric) extension 
of the classical Onsager relations,\cite{Onsager1931a,Onsager1931b} which classically only hold in the linearized 
case,\cite{Allaire2010,Looker2006} motivate us to 
recall the PNP equations from \onlinecite{Schmuck2010} written by field vectors ${\bf u}:=\ebrkts{n^+,n^-,\Phi}'$ as,
\bsplitl{
	{\bf D}_t{\bf u}
	-\pmb{\Delta}_{\mathbb{S}}{\bf u}
	= {\pmb{\rm I}}({\bf u})\,,&\qquad\textrm{for }({\bf t},{\bf x})\in\pmb{\Omega}_T
	:=[\Omega_T,\Omega_T,\Omega_T]'\,,
\\
{\bf u}(0,x)={\bf h}\,,&\qquad\textrm{in }\pmb{\Omega}
	:=[\Omega,\Omega,\Omega]'\,,
\\
{\bf u} = {\bf g}_l\,,&\qquad\textrm{on }\pmb{\Gamma}_T^l
	:= [\Gamma_T^l,\Gamma_T^l,\Gamma_T^l]'\,,
\\
{\bf u} = {\bf g}_r\,,&\qquad\textrm{on }\pmb{\Gamma}_T^r
	:= [\Gamma_T^l,\Gamma_T^l,\Gamma_T^l]'\,,
\\
(\mathbb{S}({\bf u})\pmb{\nabla}_{\bf{n}}{\bf u})^i
	:= 
	s_{ij}({\bf u})\nabla_n^j{\rm u}^j
	= 0\,,&\qquad\textrm{on }\pmb{\Gamma}_T^N:=\partial\pmb{\Omega}_T
		\setminus\pmb{\Gamma}_T^D\,,
		\qquad\textrm{for }i=1,2,3\,,
}{eq:HoCP}
where ${\bf x} := [x,x,x]'$ with $x\in\Omega\subset\mathbb{R}^N$ corresponds to the coordinate 
field for each component of the field vector ${\bf u}$ and ${\bf t}:=[t,t,t]'$ with $t\in ]0,T[$ for any $T\in\mathbb{R}_+$ is the accordingly defined time field. The variables ${\rm u}^1=n^+$, ${\rm u}^2=n^-$, and ${\rm u}^3=\Phi$ represent the concentration of positive ions, the density of negative ions, and the induced electric potential, respectively.
We further use the convention 
$\Omega_T:=]0,T[\times\Omega$. The notation $\pmb{\Omega}_T$ accounts for the fact that the components of 
the field vector ${\bf u}$ are defined in different domains of the porous medium later on, i.e., either in the whole domain $\Omega$ or only in the electrolyte phase $\Omega^s$. We further denote 
$\pmb{\Delta}_{\mathbb{S}}{\bf u}
:=\pmb{\rm div}\,\brkts{\mathbb{S}({\bf u})\pmb{\nabla}{\bf u}}$
with $\mathbb{S}({\bf u}):=\brcs{s_{i_kj_l}({\bf u})}_{\tiny{\begin{array}{c}
	 1\leq i, j\leq d \\
	1\leq k,l\leq N
	\end{array}}}$
for the field indices $1\leq i,j\leq 3$, the coordinate indices $1\leq k,l\leq N$ and $s_{i_kj_l}({\bf u})= s_{ij}({\bf u})\delta_{kl}$ with $\delta_{kl}$ 
the Kronecker symbol, $\nabla_n:={\rm n}\cdot\nabla$ with ${\rm n}$ the normal vector pointing outward of $\Omega$, $\pmb{\Gamma}_T^D:=[\Gamma_T^D,\Gamma_T^D,\Gamma_T^D]'$ with $\Gamma_T^D:=\Gamma^l_T\cup\Gamma^r_T$ and $\pmb{\Gamma}_T^N$, 
$\pmb{\Gamma}_T^\iota$ 
for $\iota=r,l$ 
are correspondingly defined, and
\bsplitl{
\quad	{\bf D}_{\bf t}
	& := \ebrkts{
	\begin{array}{ccc}
	\partial_t & 0 & 0 \\
	0 & \partial_t & 0 \\
	0 & 0 & 0
	\end{array}
	}\,,
\qquad
	\brcs{s_{ij}({\bf u})}_{
	1\leq i,j\leq 3
	}
	:= \ebrkts{
	\begin{array}{ccc}
	1 & 0 & n^+ \\
	0 & 1 & -n^- \\
	0 & 0 & \lambda^2 
	\end{array}
	}\,,
\qquad	\pmb{\nabla}
	:= \mathbb{I}\nabla
	:= \ebrkts{
	\begin{array}{ccc}
	\nabla & 0 & 0 \\
	0 & \nabla & 0 \\
	0 & 0 & \nabla 
	\end{array}
	}\,,
\\
\quad	\pmb{\rm div}
	&:= \mathbb{I}{\rm div}
	:= \ebrkts{
	\begin{array}{ccc}
	{\rm div} & 0 & 0 \\
	0 & {\rm div} & 0 \\
	0 & 0 & {\rm div} 
	\end{array}
	}\,,
\qquad
	{\bf u}
	:= \ebrkts{ n^+,n^-,\Phi}'\,,
\qquad	{\bf I}({\bf u})
	:= 
	\ebrkts{0,0, n^+-n^-}'\,,
\\
\quad	{\bf g}_l
	&:= \ebrkts{ n^+_{l},n^-_{l},\phi_{l}}'\,,
\qquad	{\bf g}_r
	:= \ebrkts{ n^+_{r},n^-_{r},\phi_{r}}'\,,
\qquad	{\bf h}
	:= \ebrkts{{\rm h}^1,{\rm h}^2,0}'\,.
}{eq:BaDe}
$\pmb{\Gamma}^\iota$ represents the Dirichlet (for $\iota=D$) and Neumann (for $\iota=N$) boundary surrounding the 
porous medium $\Omega$, i.e. $\partial\Omega=\Gamma^D\cup\Gamma^N$. Hence, the first equation \reff{eq:HoCP}$_1$ is equivalent to the following classical system,
\bsplitl{
\partial_t n^+ 
	& = {\rm div}\brkts{\nabla n^+ + n^+\nabla\Phi}\,,
\\
\partial_t n^- 
	& = {\rm div}\brkts{\nabla n^- - n^-\nabla\Phi}\,,
\\
-\lambda^2\Delta\Phi
 	& = n^+ - n^-\,,
}{PNP}
which can be interpreted as a gradient flow of the following free energy,
\bsplitl{
F = U-TS = \int\brkts{\sum_{i}{\rm u}^i({\rm log}\,{\rm u}^i - 1)
	+ \sum_iz_i{\rm u}^i\Phi - \lambda^2\av{\nabla\Phi}^2}\,dx\,.
}{FE}
We recall that \reff{FE} builds the basis of dilute solution theory which accounts for thermodynamic quantities 
such as entropy $S$ formed by the first summand in the integral \reff{FE}. The remaining integrands such as energy density of interactions (second term) and energy density of the electric field (third term) constitute to the internal energy $U$. We note that from the energy 
\reff{FE} we can obtain the chemical potential of the ion densities ${\rm u}^1=n^+$ and ${\rm u}^2=n^-$ by 
taking the first variation with respect to ${\rm u}^1$ and ${\rm u}^2$, respectively.

An interesting question is whether the minimization of the free energy \reff{FE} by a gradient 
flow also follows the physically relevant path far from thermodynamic equilibrium. In which physical 
sense does the flow with respect to the Wasserstein distance\cite{Ambrosio2004,Gianazza2009,Herrmann2010,Jordan1998} provide optimality?

\section{The microscopic porous medium formulation and main results}\label{sec:pmPNP} 
The study in this article relies on the system \reff{eq:HoCP} reformulated for periodic porous media. 
A scaling parameter $s$ is defined as the ratio between the microscopic length 
scale $\ell$ and the macroscopic size $L$ of the porous medium, i.e. $s:=\frac{\ell}{L} \ll 1$. It is assumed that $s$ scales 
the periodicity of the reference cell $Y\subset\mathbb{R}^N$ which defines the micro-geometry. 
In this reference cell, we denote 
the fluid (liquid) region by $Y^s\subset Y$ such that its complement is the solid phase. After periodically covering 
the domain $\Omega$ by such cells $Y$, we denote the resulting macroscopic domain 
of the periodic union of the subsets $Y^s$ by $\Omega^s$ and its complement by 
$B^s:=\Omega\setminus\Omega^s$.
Hence, the perforated domain $B^s$ represents the solid phase and $\Omega^s$ the liquid phase, 
see Figure \ref{fig:MicrMacr}.

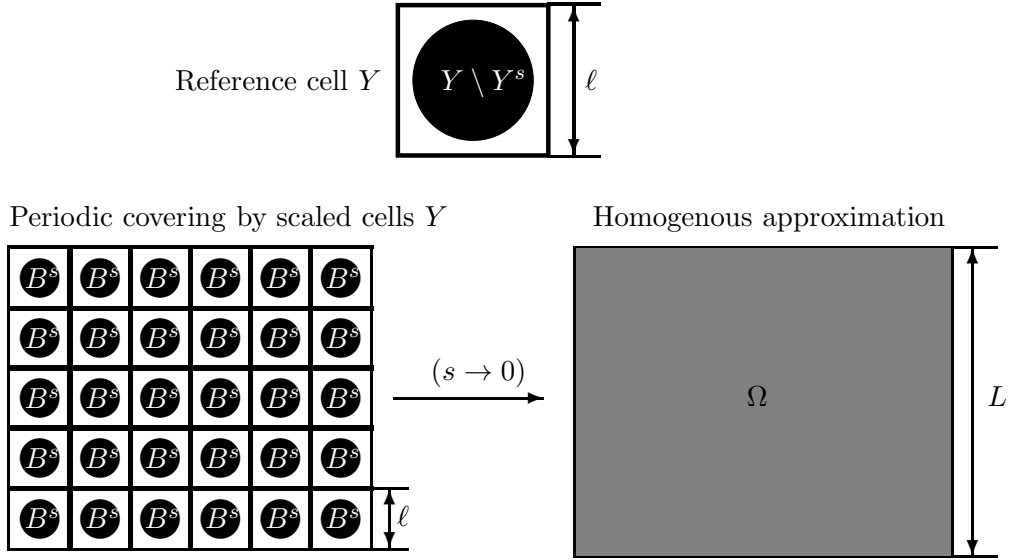
\begin{figure}[htbp]
\begin{center}
\setlength{\unitlength}{1cm}
\begin{picture}(13,7.3)(0,0)					

\setlength{\unitlength}{0.1cm}
\setlength{\unitlength}{1cm}
\thicklines
\multiput(0.0,0.0)(0.8,0.0){6}{\framebox(0.8,0.8)[s]{\put(0.4,0.0){\circle*{0.6}}}}
\multiput(0.0,0.0)(0.8,0.0){6}{\framebox(0.8,0.8)[s]{\put(0.2,0.0){\textcolor{white}{$B^s$}}}}
\multiput(0.0,0.8)(0.8,0.0){6}{\framebox(0.8,0.8)[s]{\put(0.4,0){\circle*{0.6}}}}
\multiput(0.0,0.8)(0.8,0.0){6}{\framebox(0.8,0.8)[s]{\put(0.2,0){\textcolor{white}{$B^s$}}}}
\put(4.8,-0.015){\line(1,0){0.5}}
\multiput(0.0,1.6)(0.8,0.0){6}{\framebox(0.8,0.8)[s]{\put(0.4,0){\circle*{0.6}}}}
\multiput(0.0,1.6)(0.8,0.0){6}{\framebox(0.8,0.8)[s]{\put(0.2,0){\textcolor{white}{$B^s$}}}}
\put(4.8,0.8){\line(1,0){0.5}}
\put(5.05,0.0){\vector(0,1){0.8}}
\put(5.05,0.8){\vector(0,-1){0.8}}
\put(5.15,0.3){$\ell$}
\multiput(0.0,2.4)(0.8,0.0){6}{\framebox(0.8,0.8)[s]{\put(0.4,0){\circle*{0.6}}}}
\multiput(0.0,2.4)(0.8,0.0){6}{\framebox(0.8,0.8)[s]{\put(0.2,0){\textcolor{white}{$B^s$}}}}
\multiput(0.0,3.2)(0.8,0.0){6}{\framebox(0.8,0.8)[s]{\put(0.4,0){\circle*{0.6}}}}
\multiput(0.0,3.2)(0.8,0.0){6}{\framebox(0.8,0.8)[s]{\put(0.2,0){\textcolor{white}{$B^s$}}}}
\put(0.0,4.3){Periodic covering by scaled cells $Y$}
\put(5.6,2.25){($s\to 0$)}
\put(5.1,2.0){\vector(2,0){2.0}}
\put(7.5,0.0){\fcolorbox{black}{grau}{\makebox(4.8,3.9)[]{}}} %
\put(7.75, 4.3){Homogenous approximation}
\put(9.8, 1.9){$\Omega$}

\put(12.53,4.005){\line(1,0){0.6}}
\put(12.53,-0.105){\line(1,0){0.6}}
\put(12.8,2){\vector(0,1){2.0}}
\put(12.8,2.0){\vector(0,-1){2.1}}
\put(13.0,1.9){$L$}

\put(2.2,6.1){Reference cell $Y$}

\put(7.15,7.245){\line(1,0){0.7}}
\put(7.15,5.215){\line(1,0){0.7}}
\put(7.5,6.2){\vector(0,1){1.025}}
\put(7.5,6.25){\vector(0,-1){1.02}}
\put(7.65,6.1){$\ell$}
\put(5,5.2){
\begin{tikzpicture}
\draw[fill=black] (15,15) circle (.8cm); 
\draw[line width=0.6mm] (14,14) rectangle (16,16); 
\end{tikzpicture}
}
\put(5.73,6.1){\textcolor{white}{$Y\setminus Y^s$}}

\end{picture}
\caption{{\bf Left:} Strongly heterogeneous material as a periodic covering of reference cells $Y:=[0,\ell]^N$. 
{\bf Top, middle:} The reference cell $Y$ represents a characteristic mean pore geometry. 
{\bf Right:} The ``homogenization limit'' $s:=\frac{\ell}{L}\to 0$ stands here for the leading order approximation of non-standard two-scale expansions.} 
\label{fig:MicrMacr}
\end{center}
\end{figure}

Under these considerations, the material tensor $\mathbb{S}$ from \reff{eq:HoCP}$_2$ 
depends now on $s$ too, i.e., 
\bsplitl{
\brcs{s_{i_kj_l}^s({\bf x},{\bf u}_s)}_{1\leq i,j\leq 3}
	:=\ebrkts{
	\begin{array}{ccc}
	1 & 0 & n^+_s \\
	0 & 1 &  -n^-_s \\
	0 & 0 & \epsilon(x/s) 
	\end{array}
	}\delta_{kl}\,,\quad\textrm{for }1\leq k,l\leq N\,,
}{eq:epsMaTe}
where $\epsilon(x):= \lambda^2\chi_{\Omega^s}(x)+\alpha\chi_{B^s}(x)$ with the classical dimensionless 
Debye length $\lambda:=\frac{\lambda_D}{L}$ of the PNP system (\ref{PNP}), $\alpha=\frac{\epsilon_m}{\epsilon_f}$ 
is the dimensionless dielectric permittivity where $\epsilon_m$ and $\epsilon_f$ are the dielectric 
permittivities of the solid and liquid phase, respectively. In \reff{eq:epsMaTe} one recognizes 
that also the physical quantities like concentrations $n^+_s,\, n^-_s$ and electric potential 
depend on the scaling parameter $s$. Hence, the problem \reff{eq:HoCP} reads now in the periodic setting as follows,
\bsplitl{
{\bf D}_t{\bf u}_s
	-\pmb{{\rm div}}\,\brkts{\mathbb{S}^s({\bf u}_s)\pmb{\nabla}{\bf u}_s} 
	= {\pmb{\rm I}}({\bf u}_s)\,,
	&\qquad\textrm{in }\pmb{\Omega}^s_T
	:=[\Omega^s_T,\Omega^s_T,\Omega_T]'\,,
\\
{\bf u}_s(0,x)={\bf h}_s\,,&\qquad\textrm{in }\pmb{\Omega}^s
	:=[\Omega^s,\Omega^s,\Omega]'\,,	
\\
	\mathbb{S}^s({\bf u}_s)\pmb{\nabla}_{\bf{n}}{\bf u}_s
	= {\bf 0}\,,&\qquad\textrm{on }\pmb{\Gamma}^N_T\,,
\\
{\bf u}_s = {\bf g}_\iota\,,&\qquad\textrm{on }\pmb{\Gamma}_T^\iota\,,
	\quad\iota=l,r\,,
\\
\brkts{
	\mathbb{S}^s({\bf u}_s)\pmb{\nabla}_{\bf{n}}{\bf u}_s
	}{\bf f}_i
	= 0\,,&\qquad\textrm{on }I^s_T\,,
	\quad\textrm{for }i=1,2\,,
\\
\brkts{
	\mathbb{S}^s({\bf u}_s)\pmb{\nabla}_{\bf{n}}{\bf u}_s
	}{\bf f}_3
	\bigl|_{\Omega^s_T}
	= 
	\brkts{
	\mathbb{S}^s({\bf u}_s)\pmb{\nabla}_{\bf{n}}{\bf u}_s
	}{\bf f}_3
	\bigl|_{B^s_T}
	\,,&\qquad\textrm{on }I^s_T\,,
\\
{\rm u}_s^3\bigl|_{\Omega^s_T} 
	= {\rm u}_s^3\bigl|_{B^s_T}
	\,,&\qquad\textrm{on }I^s_T\,,			
}{eq:PeHoPr}
where ${\bf f}_i:=\ebrkts{\delta_{i1},\delta_{i2},\delta_{i3}}'$ for $i=1,2,3$ and $\delta_{ij}$ is the Kronecker delta and 
$I^s_T:=\partial\Omega^s_T
		\setminus\brcs{\Gamma_T^D\cup\Gamma_T^N}$ the solid-electrolyte interface.
From \reff{eq:PeHoPr} it follows that the flux with respect to ${\bf u}$ is in general not differentiable. This motivates to study \reff{eq:PeHoPr} in the sense of weak solutions. Moreover, 
the main difficulty and difference of this work is the nonlinear structure which prevents the material 
tensor $\mathbb{S}^s({\bf u}_s)$ 
to be a strongly elliptic operator. For convenience, we rewrite 
\reff{eq:PeHoPr} for $[n^+_s,n^-_s,\Phi_s] =[{\rm u}^1_s,{\rm u}^2_s,{\rm u}^3_s]$ 
in the classical form here, that is,
\bsplitl{
\partial_t n^+_s
	= {\rm div}\brkts{\nabla n^+_s + n^+_s\nabla\Phi_s}
	&\qquad\textrm{in }\Omega^s_T\,,
\\
\partial_t n^-_s
	= {\rm div}\brkts{\nabla n^-_s - n^-_s\nabla\Phi_s}
	&\qquad\textrm{in }\Omega^s_T\,,
\\
-{\rm div}\brkts{\epsilon(x/s)\nabla\Phi_s}
	= n^+_s - n^-_s
	&\qquad\textrm{in }\Omega_T\,,
}{clPNPper}
where the corresponding ionic fluxes are defined by,
\bsplitl{
{\rm j}^+_s 
	& := \nabla n^+_s + n^+_s\nabla\Phi_s\,,
\\
{\rm j}^-_s 
	& := \nabla n^-_s - n^-_s\nabla\Phi_s\,,	
}{Flpm}
and the boundary and initial conditions are,
\bsplitl{
n^+_s = {\rm h}^1_s\,\quad n^-_s = {\rm h}^2_s
	& \qquad\textrm{in }\Omega^s\,,
\\
n^+_s = g_\iota\,\quad n^-_s = g_\iota
	& \qquad\textrm{on }\Gamma^\iota_T\,,	
\\
{\rm j}^+_s\cdot{\rm n} = 0\,\quad {\rm j}^-_s\cdot{\rm n} = 0
	&\qquad\textrm{on }\Gamma^N_T\,,
}{ICBCcl}
where $\iota=l,r$. Our main result relies on the assumption of local 
thermodynamic equilibrium, which is widely used and generally accepted.\cite{Bennethum2000,Nelson1999,Schmuck2012,Schmuck2012zamm}

\medskip

\begin{defn}\label{def:LoEq} \emph{(Local thermodynamic equilibrium)} We say that the reference cells $Y$ 
are in local thermodynamic equilibrium if and only if for each $x/s=y$ of the same reference cell $Y$ it holds that
\bsplitl{
\mu_0^r := {\rm log}\,{\rm u}^r(t,x) +z_r{\rm u}^3(t,x)
	= {\rm const.}\,,
}{LoEqCo}
where $\mu_0^r$ denotes a constant value of the chemical potential of positive ($r=1$) 
and negative ($r=2$) ion densities. Hence, the locally constant potential $\mu_0^r$ only assumes different values in different reference cells.
\end{defn}

\medskip

In order to simplify the subsequent derivations, we make the following 

\medskip

{\bf Assumption (ABC):} (Academic Boundary Conditions) \emph{We assume for the ion densities ${\rm u}^r_s$ for $r=1,2$ 
homogeneous Dirichlet boundary conditions on $\partial\Omega$ and no-flux boundary conditions on the 
solid-electrolyte interface $I^s:=\partial\Omega^s\setminus\brcs{\Gamma^D\cup\Gamma^N}$ inside the porous medium $\Omega$. 
For the electric potential ${\rm u}^3_s=\Phi_s$, we assume homogeneous Neumann boundary conditions on the boundary of 
the porous medium, i.e., $\partial\Omega$, and continuity, i.e., $\Phi\bigr|_{\Omega^s}=\Phi\bigr|_{B^s}$, at the solid-electrolyte interface $I^s:=\partial\Omega^s\setminus\brcs{\Gamma^D\cup\Gamma^N}$ as well as continuity of the corresponding fluxes, i.e.,
$
\epsilon(x/s)\nabla_n\Phi_s\bigr|_{\Omega^s}
	= \epsilon(x/s)\nabla_n\Phi_s\bigr|_{B^s}
$
on $I^s\,.$
}

\medskip

These considerations allow us to state our main result which consists of passing to 
the limit $s\to 0$. This limit has the physical meaning of homogeneously mixing the solid and 
the liquid phase in the porous material under constant 
volume fractions, see Figure \ref{fig:MicrMacr}. Such a homogeneous description of a porous medium is a good 
approximation if the medium is very heterogeneous, i.e., $s=\frac{\ell}{L}$ is very small.

\begin{thm}\label{thm:AsExp} We assume that the reference cells $Y$ are in local thermodynamic 
equilibrium. The boundary conditions satisfy the {\rm Assumption (ABC)}. Then, the solution ${\bf u}_s({\bf t},{\bf x}) := [{\rm u}_s^1(t,x),{\rm u}_s^2(t,x),{\rm u}_s^3(t,x)]Õ$ of problem \reff{eq:PeHoPr} 
admits the following formal asymptotic expansions,
\bsplitl{
{\rm u}_s^r
	& = {\rm u}_0^r
	- s \sum_{k=1}^N \xi^{r_k}(t,x,x/s)\frac{\partial {\rm u}_0^3}{\partial x_k}
	+ s^2 \sum_{k,l=1}^N\zeta^{r_{kl}}(t,x,x/s){\rm u}_0^r
	+\dots
	\quad\textrm{for }r=1,2\,,
\\
{\rm u}_s^3
	& = {\rm u}_0^3
	- s \sum_{k=1}^N \xi^{3_k}(x/s)\frac{\partial {\rm u}_0^3}{\partial x_k}
	+ s^2 \sum_{k,l=1}^N\zeta^{3_{kl}}(x/s)\frac{\partial^2 {\rm u}_0^3}{\partial_{x_k}\partial_{x_l}}
	+\dots\,,
}{AsExp}
where $\xi^{r_{k}}(\cdot,\cdot,y)\in V(\Omega_T,W_\sharp(Y^s))$ solves \reff{xir}, 
$\xi^{3_k}(y)\in W_\sharp(Y)$ solves \reff{xi3}, $\zeta^{r_{kl}}(\cdot,\cdot,y)\in V(\Omega_T,W_\sharp(Y^s))$ solves 
\reff{ExUnZetar}, $\zeta^{3_{kl}}(y)\in W_\sharp(Y)$ solves 
\reff{ExUnZeta3}, and ${\bf u}_0$ is a solution of the following upscaled system,
\bsplitl{
\begin{cases}
p\partial_t {\rm u}^r_0
	-p\Delta {\rm u}^r_0 
	+{\rm div}\brkts{\mathbb{D}^r(t,x) \nabla{\rm u}^3_0}
	-{\rm div}\brkts{z_r{\rm u}_0^r\mathbb{M}\nabla{\rm u}^3_0}
	=0
	&\textrm{in }\Omega_T\,,
\\
-{\rm div}\brkts{\pmb{\epsilon}^0\nabla {\rm u}_0^3}
	=p\brkts{ {\rm u}_0^1-{\rm u}_0^2}
	&\textrm{in }\Omega_T\,,
\end{cases}
}{UpScRe}
where $p:=\av{Y^s}/\av{Y}$ is the porosity.
The tensor $\mathbb{D}^r(t,x):=\brcs{{\rm D}^r_{kl}(t,x)}_{1\leq k,l\leq N}$ related to diffusion, 
the tensor $\mathbb{M}:=\brcs{{\rm M}_{kl}}_{1\leq k,l\leq N}$ related to electro-convection, and 
the effective permeability tensor $\pmb{\epsilon}^0:=\brcs{\epsilon^0_{kl}}_{1\leq k.l\leq N}$
are defined by,
\bsplitl{
{\rm D}^r_{ik}(t,x)
	& := \frac{1}{\av{Y}}
	\int_{Y^s}\sum_{j=1}^N\brcs{\delta_{ij}\partial_{y_j}\xi^{r_k}(t,x,y)}\,dy\,,
\\
{\rm M}_{ik}
	&:= \frac{1}{\av{Y}}
	\int_{Y^s}\sum_{j=1}^N\brcs{\delta_{ik}-\delta_{ij}\partial_{y_j}\xi^{3_k}(y)}\,dy\,,
\\
\epsilon^0_{ik}
	:&=
	\frac{1}{\av{Y}}
	\int_Y \sum_{j=1}^N
			\epsilon(y)\brkts{
				\delta_{ik}-\delta_{ij}\partial_{y_j}\xi^{3_k}(y)
			}		
	\,dy\,.
}{Te}
The second order correctors $\zeta^{r_{kl}}$ only exist for positive 
densities ${\rm u}_0^r$.
\end{thm}

We note that for the classical asymptotic two-scale expansions\cite{Bensoussans1978,Cioranescu2000} an upscaling 
is performed in \onlinecite{Schmuck2012zamm} where also error erstimates between the microscopic periodic formulation and the upscaled equations are derived.

\begin{rem}\label{rem:GePmPNP}
{\rm (1)} Theorem \ref{thm:AsExp} is an extension of the two-scale convergence results in \onlinecite{Schmuck2010} 
by the non-classical asymptotic expansion \reff{AsExp}$_1$. We note that the expansions \reff{AsExp} are 
only formal because convergence of such 
series is a priori not guaranteed and possible boundary layers are 
neglected.\\
{\rm (2)} The effect of the upscaling in the above theorem can be best seen in the change of the material tensor 
\reff{eq:epsMaTe}, which reads for the new system \reff{UpScRe} as follows
\bsplitl{
\brcs{s_{i_kj_l}^s({\bf x},{\bf u}_0)}_{1\leq i,j\leq 3}
	:=\ebrkts{
	\begin{array}{ccc}
	p\delta_{kl} & 0 & -{\rm D}^1_{kl}+z_r{\rm u}^1_0{\rm M}_{kl} \\
	0 & p\delta_{kl} &  -{\rm D}^1_{kl}-z_r{\rm u}^2_0{\rm M}_{kl} \\
	0 & 0 & \epsilon^0_{kl}(x/s) 
	\end{array}
	}\,,\quad\textrm{for }1\leq k,l\leq N\,.
}{UpScMaTe}
A comparison of \reff{UpScMaTe} with \reff{eq:epsMaTe} clearly motivates the use of the term ``material tensor'' 
in the context of porous or composite media. \\
{\rm (3)} The effective material tensor \reff{UpScMaTe} can also be considered as a generalized effective, concentration dependent conductivity tensor (as in heat/diffusion equations).
\end{rem}

We note that the effective macroscopic material tensor \reff{UpScMaTe} reveals that one can controle and enhance material transport by 
adjusting the contrast in the electric permittivities between the electrolyte and the porous medium. The different 
upscaling ideas as well as different physical and mathematical assumptions play an important role for an improved understanding of how different microscopic material properties such as pore geometries, electric permittivities, charge numbers 
of the ions, and surface charges of the porous medium influence the transport characteristics on the macroscale.

Our last result guarantees the existence and uniqueness of solutions for the porous media 
PNP equations \reff{UpScRe}. In order to keep the 
presentation clear and to prevent a technical derivation, we restrict our considerations 
to purely academic boundary conditions in the next 

\medskip

\begin{lem}\label{lem:ExReMo}
Let $\Omega\subset\mathbb{R}^N$ be a domain with smooth boundary 
$\partial\Omega$ and $T>0$ small enough. The boundary conditions satisfy {\rm Assumption (ABC)}. 
Then, the coupled system of equations,
\bsplitl{
\begin{cases}
p\partial_t {\rm u}^r
	-p\Delta {\rm u}^r
	+{\rm div}\brkts{\mathbb{D}^r(t,x) \nabla{\rm u}^3}
	-{\rm div}\brkts{z_r{\rm u}_0^r\mathbb{M}\nabla{\rm u}^3}
	=0
	&\textrm{in }\Omega_T\,,
\\
-{\rm div}\brkts{\pmb{\epsilon}^0\nabla {\rm u}^3}
	=p\brkts{ {\rm u}^1-{\rm u}^2}
	&\textrm{in }\Omega_T\,,
\end{cases}
}{ReUpPNP}
together with the following initial and boundary conditions,
\bsplitl{
{\rm u}^r(0,x)
	& \in L^2(\Omega)\cap H^1_0(\Omega)
	\qquad\qquad\textrm{for }r=1,2\,,
\\
{\rm u}^r 
	& = 0 \qquad\textrm{on }\partial\Omega\times ]0,T[\quad\,\,\;\;\;\;\textrm{ for }r=1,2\,,
\\
\nabla_n{\rm u}^3 
	& = 0 \qquad\textrm{on }\partial\Omega\times ]0,T[\,,
}{ICBC}
has unique weak solutions ${\rm u}^r \in V(\Omega_T):=L^\infty(]0,T[;L^2(\Omega))\cap L^2(]0,T[;H^1_0(\Omega))$ 
and ${\rm u}^3\in L^\infty(]0,T[;H^2(\Omega))$. That means, ${\rm u}^r(t,x)$, $r=1,2$, and 
${\rm u}^3(t,x)$ solve
\bsplitl{
\begin{cases}
p\frac{d}{dt}\brkts{{\rm u}^r,\varphi^r}
	+p\brkts{\nabla {\rm u}^r,\nabla\varphi^r}
	-\brkts{{\mathbb D}^r(t,x)\nabla {\rm u}^3, \nabla\varphi^r}
	&
\\\qquad\qquad\qquad\qquad\qquad
	+\brkts{z_r{\rm u}^r{\mathbb M}\nabla {\rm u}^3, \varphi^r}
	= 0
	&\forall\varphi^r\in H^1_0(\Omega)\,,
\\
\brkts{\pmb{\epsilon}^0\nabla{\rm u}^3, \nabla\varphi^3}
	= p\brkts{{\rm u}^1-{\rm u}^2,\varphi^3}
	&\forall\varphi^3\in\ol{H}^1(\Omega)\,,
\end{cases}
}{WUpPNP}
where $\ol{H}^1(\Omega):=\brcs{u\in H^1(\Omega)\,\bigr|\, \int_\Omega u\,dx = 0}$.
\end{lem}

\medskip

From a rigorous point of view, Lemma \ref{lem:ExReMo} finally guarantees that the second order 
terms in the asymptotic expansion \reff{AsExp} are locally well-defined.


%% file: physics.tex
In this section, we state the physical setting which leads to strongly oscillating electric potentials in 
composites such as a porous medium permeated by a dilute electrolyte. A related example where 
such oscillations are well-known is the electric field over a material with strongly heterogeneous conductivities. This is also one of the classical fields of effective media theory\cite{Torquato2002} and homogenization theory\cite{Bensoussans1978,Cioranescu2000} where one often assumes a periodic representation of the heterogeneities for simplicity.
In fact, the high-frequency electric permittivity and the low-frequency electric conductivity are formally equivalent because of the 
equivalence in the governing equations. However, 
the situation for the PNP system here is slightly different since we have to deal with a non-linearly 
coupled system of equations. As explained previously, we account for this difference by a non-standard asymptotic 
expansion that factors the strong influence of the electric potential in. Moreover, the equivalence between permittivity and conductivity implies that their mathmatical computation is equivalent.

As in the case of conductivity, the strength of the oscillations can be controlled by the distance between the 
different electric permittivities for our problem here. Since we study a dilute electrolyte, we can expect an electric 
permittivity of the liquid phase to be around $80$ at room temperature and a frequency under $1kHz$ 
(of course, this also depends on the electrolyte employed). For the solid 
phase, we can expect an electric permittivity between $2$ and $5$, i.e., paper $3$,
alumina $4.5$,
teflon $2.1$,
porcelain $5.1$, and 
plexiglas $2.8$. But in many fields, a systematic derivation of effective media quantities such as the electric permittivity is still lacking.\cite{Brovelli2008} Our subsequently derived equations reliably 
define such an effective electric permittivity for periodic porous media. We emphasize the importance of characterizing 
porous materials with respect to dielectric properties in microelectronics.\cite{Maex2003} Moreover, we believe that a 
systematic and reliable upscaling of such 
complex composites using geometric and material properties together with experimental validation gives 
promising perspectives for new scientific, technological, and industrial applications.

%% file: notation.tex
The following classical exposition recalls central definitions and results from \onlinecite{Bensoussans1978,Cioranescu2000,Pavliotis2008}.
For the microscopic variable $y := \frac{x}{s}$, we obtain the following relation for 
gradients applied to functions 
$\psi_s(x):=\psi\brkts{x,\frac{x}{s}}$, i.e.,
\bsplitl{
\nabla \psi_s(x)
	= \nabla_x\psi(x, y)
	+\frac{1}{\epsilon}\nabla_y\psi(x,y)\,.
}{nabla}

\medskip

Homogeneous Neumann problems for Poisson equations for example require the use of
the quotient space $W_\sharp(Y):=H^1_\sharp(Y)/\mathbb{R}$. This space of 
equivalence classes is defined with respect to the relation,
\bsplitl{
u \simeq v\quad \Leftrightarrow\quad  u-v \textrm{ is a constant, }\forall u,v\in H^1_\sharp(Y)\,.
}{Eqv}
For notational brevity, we do not introduce additional notation for an element of this equivalence 
class. A representative element of this equivalence class can be chosen by the 
following mean zero condition, that means,
\bsplitl{
W_\sharp(Y)
	:= \brcs{u\in H^1_\sharp(Y)\,\bigr|\, \mathcal{M}_Y(u)=0}\,,
}{Wp}
where,
\bsplitl{
\mathcal{M}_Y(u)
	:= \frac{1}{\av{Y}}\int_Y u(y)\,dy\,.
}{MY}

\medskip

\begin{lem}\label{lem:WpDual}
The following quantitiy, 
\bsplitl{
\N{u}{W_\sharp}
	= \N{\nabla u}{L^2(Y)}\qquad\forall u,v\in W_\sharp(Y)\,,
}{WpD}
defines a norm on $W_\sharp(Y)$. Moreover, the dual space $(W_\sharp(Y))'$ can be 
identified by the set,
\bsplitl{
\brcs{F\in (H^1_\sharp(Y))Õ\,\bigr|\, F(c)=0\quad \forall c\in\mathbb{R}}\,,
}{DS}
with,
\bsplitl{
\abrkts{F, u}_{(W_\sharp(Y))',W_\sharp(Y)}
	= \abrkts{F, u}_{(H^1_\sharp(Y))',H^1_\sharp(Y)}
	\quad\forall u\in W_\sharp(Y)\,.
}{DNPr}
\end{lem}

\medskip

\begin{defn}\label{def:Mab}
Let $c,C\in\mathbb{R}$ be such that $0<c<C$ and let $D\subset\mathbb{R}^N$. We call 
$N\times N$ matrices $A=\brcs{a_{ij}}_{1\leq i,j \leq N} \in (L^\infty(U))^{N\times N}$ 
strongly elliptic, if for any $u\in\mathbb{R}^N$ and a.e. in $D$ it holds that,
\bsplitl{
(A(x)u, u) 
	& \geq c\av{u}^2\,,
\\ 
\av{A(x)u} 
	& \leq C \av{u}\,.
}{StrEll}
\end{defn}

\medskip

In our analysis we mainly have to deal with $A=\brcs{\delta_{ij}}_{1\leq i,j \leq N}$ which obvously 
satisfies the conditions of Definition \ref{def:Mab}.

\medskip

\begin{thm}\label{thm:ElEx}
Let $A$ be a strongly elliptic matrix with $Y$-periodic 
coefficients and $f\in (W_\sharp(Y))'$. Then the problem,
\bsplitl{
\begin{cases}
\textrm{Find $u\in W_\sharp(Y)$ such that}\\
(A\nabla u,\nabla v)_Y = (f,v)
	\qquad \forall v\in W_\sharp(Y)\,.
\end{cases}
}{EP}
has a unique solution. Moreover,
\bsplitl{
\N{u}{W_\sharp(Y)}
	\leq \frac{1}{c}\N{f}{(W_\sharp(Y))Õ}\,.
}{EPEs}
\end{thm}

\begin{rem}
Since Theorem \ref{thm:ElEx} makes a uniqueness statement, we consider in this case the space 
$W_\sharp(Y)$ defined in \reff{Wp}. We apply this convention in the whole article.
\end{rem}

We frequently use the following space,
\bsplitl{
V(\Omega_T)
	:= L^\infty(]0,T[; L^2(\Omega))\cap L^2(]0,T[;H^1(\Omega))\,.
}{}

%% file: proof.tex
We first rewrite the second term in equation \reff{eq:PeHoPr}$_1$ with help of the property 
\reff{nabla} in the following way
\bsplitl{
{\bf div}\brkts{
		\mathbb{S}^s(\pmb{\psi}_s,{\bf x})
		\pmb{\nabla}\pmb{\psi}_s
	}
	& = s^{-2}{\bf div}_y\brkts{\mathbb{S}\brkts{
			\pmb{\psi}_s,{\bf y}
		}\pmb{\nabla}_y\pmb{\psi}_s}
\\&		
	+s^{-1}\brcs{
		{\bf div}_x\brkts{
			\mathbb{S}(\pmb{\psi}_s,{\bf y})\pmb{\nabla}_y\pmb{\psi}_s
		}
		+{\bf div}_y\brkts{
			\mathbb{S}(\pmb{\psi}_s,{\bf y})\pmb{\nabla}_x\pmb{\psi}_s
		}
	}
\\&
	+ {\bf div}_x\brkts{
		\mathbb{S}\brkts{
		\pmb{\psi}_s,{\bf y}
		}\pmb{\nabla}_x
	}
\\ &:=
	\ebrkts{
		s^{-2}\mathcal{S}_0 
		+ s^{-1}\mathcal{S}_1
		+\mathcal{S}_2
	}\pmb{\psi}\brkts{{\bf x},\frac{{\bf x}}{s}}\,,
}{NlLpl}
where $\pmb{\psi}_s({\bf x})$ is an arbitrary function as in \reff{nabla}. For the subsequent 
considerations we apply the following notation,
\bsplitl{
\ebrkts{\mathcal{S}_2}^N\pmb{\psi}_s
	& := \ebrkts{
		{\bf div}\brkts{\mathbb{S}(\pmb{\psi}_s,{\bf y})\pmb{\nabla}_y \pmb{\psi}_s}
	}^N
	:= \ebrkts{
		{\bf div}\brkts{\mathbb{S}(\pmb{\psi}_s,{\bf y})\pmb{\nabla}_y \pmb{\psi}_s}
	}^N
\\&
	:= \brcs{
		{\rm div}_{x}\brkts{{\rm s}_{ij}(\pmb{\psi},{\bf y})\nabla_{x} \psi^j_s}{\bf f}_i
	}_{(i=1,j=3)\cup(i=2,j=3)}
\\&
	= {\rm div}_x(\psi_s^1\nabla_x\psi^3_s){\bf f}_1
	-{\rm div}_x(\psi_s^2\nabla_x\psi_s^3){\bf f}_2\,.
}{prNl}
Hence, the operator $[\cdot]^N$ represents a restriction to nonlinear components given by the 
indices $(i=1,j=3)$ and $(i=2,j=3)$.

\medskip

We make now the formal Ansatz of the asymptotic expansion
\bsplitl{
{\bf u}_s({\bf t},{\bf x})
	\approx
	{\bf u}_0({\bf t},{\bf x},{\bf x}/s)
		+s{\bf u}_1({\bf t},{\bf x},{\bf x}/s)
		+s^2 {\bf u}_2({\bf t},{\bf x},{\bf x}/s)
		+\dots\,,
}{AsEx}
with ${\bf u}_i({\bf t},{\bf x},{\bf y})$ for $i=0,1,2,\dots$ such that
\bsplitl{
\begin{cases}
{\bf u}_i({\bf t},{\bf x},{\bf y}) & \textrm{is defined for $({\bf t},{\bf x})\in\pmb{\Omega}^s_T$ and 
	${\bf y}\in{\bf Y}:=[Y^s,Y^s,Y]'$,}
\\
{\bf u}_i(\cdot,\cdot,{\bf y}) & \textrm{is ${\bf Y}$-periodic.}
\end{cases}
}{ui}
The above Ansatz is formal because there is no guarantee that the series \reff{AsEx} is finite. 

\medskip

{\bf (1) Problem for terms of order $\mathcal{O}(s^{-2})$:} After inserting \reff{AsEx} into \reff{eq:PeHoPr}$_1$, using \reff{NlLpl} and \reff{prNl} gives a sequence of 
problems by equating terms with equal power in $s$, that means,
\bsplitl{
\mathcal{O}(s^{-2}):\qquad
\begin{cases}
-\mathcal{S}({\bf u}_0,{\bf y}){\bf u}_0
	= 0
	& \qquad\textrm{in }{\bf Y}\,,
\\
{\bf u}_0(\cdot,{\bf y})
	\textrm{ ${\bf Y}$-periodic\,.}&
\end{cases}
}{s-2}
If we use definition \reff{NlLpl}, then we can rewrite \reff{s-2}$_1$ as the following equation,
\bsplitl{
-{\bf div}\brkts{
	\mathbb{S}({\bf u}_0,{\bf y})\pmb{\nabla}_{y}{\bf u}_0
	}
	={\bf 0}
	\qquad\textrm{in }{\bf Y}\,,
}{1s-2}
wich is equivalent to the system,
\bsplitl{
\mathcal{O}(s^{-2}):\qquad
\begin{cases}
-\Delta_y {\rm u}^1_0
	-{\rm div}_y\brkts{{\rm u}_0^1\nabla_y{\rm u}^3_0}
	=0
	& \qquad\textrm{in } Y^s\,,
\\
-\Delta_y {\rm u}^2_0
	+{\rm div}_y\brkts{{\rm u}_0^2\nabla_y{\rm u}^3_0}
	=0
	& \qquad\textrm{in } Y^s\,,
\\
-{\rm div}_y\brkts{
		\kappa(y)\nabla_y {\rm u}^3_0
	}
	=0
	&\qquad\textrm{in } Y\,.
\end{cases}
}{2s-2}
One recognizes immediately that solvability must be first established for equation \reff{2s-2}$_3$. This is immediately 
achieved by Theorem \ref{thm:ElEx}. Moreover, this theorem implies that ${\rm u}_0^3(x,y)$ is invariant (constant) in 
$y\in Y$ as a solution of \reff{2s-2}$_3$, i.e.,
\bsplitl{
{\rm u}_0^3(t,x,y)
	= {\rm u}_0^3(t,x)\,.
}{u03}
Using invariance \reff{u03} in equations \reff{2s-2}$_1$ and \reff{2s-2}$_2$ implies with 
Theorem \ref{thm:ElEx} the additional invariances
\bsplitl{
{\rm u}_0^1(t,x,y)
	= {\rm u}_0^1(t,x)
	\quad\textrm{and}\quad
{\rm u}_0^2(t,x,y)
	= {\rm u}_0^2(t,x)\,.
}{u01u02}
Let us go over to the next problem in the sequence of equal power in $s$.

\medskip

{\bf (2) Problem for terms of order $\mathcal{O}(s^{-1})$:} (Reference cell problems) The second problem 
has the form,
\bsplitl{
\mathcal{O}(s^{-1}):\quad
\begin{cases}
-\mathcal{S}_0({\bf u}_0,{\bf y}){\bf u}_1
	= \ebrkts{\mathcal{S}_0({\bf u}_1,{\bf y})}^N{\bf u}_0
		+\mathcal{S}_1({\bf u}_0,{\bf u}_0,{\bf  y}){\bf u}_0
	\qquad\textrm{in }\pmb{Y}\,,
\\
{\bf u}_1(\cdot,{\bf u})
	\textrm{ $Y$-periodic}\,.
\end{cases}
}{s-1}
Let us write \reff{s-1} in a more intuitive form by its single components, i.e.,
\bsplitl{
\mathcal{O}(s^{-1}):\quad
\begin{cases}
-\Delta_y {\rm u}^1_1
	 -{\rm div}_y\brkts{{\rm u}_0^1\nabla_y {\rm u}_1^3}
	= {\rm div}_y\brkts{{\rm u}_1^1\nabla_y{\rm u}_0^3}
	+{\rm div}_x\nabla_y{\rm u}_0^1&
\\\qquad\qquad
	+{\rm div}_x\brkts{{\rm u}_0^1\nabla_y{\rm u}_0^3}
	+{\rm div}_y\nabla_x {\rm u}_0^1
	+{\rm div}_y\brkts{{\rm u}_0^1\nabla_x{\rm u}^3_0}
	&\quad\textrm{in }Y^s\,,
\\
-\Delta_y {\rm u}^2_1
	 -{\rm div}_y\brkts{{\rm u}_0^2\nabla_y {\rm u}_1^3}
	=- {\rm div}_y\brkts{{\rm u}_1^2\nabla_y{\rm u}_0^3}
	+{\rm div}_x\nabla_y{\rm u}_0^2&
\\\qquad\qquad
	-{\rm div}_x\brkts{{\rm u}_0^2\nabla_y{\rm u}_0^3}
	+{\rm div}_y\nabla_x {\rm u}_0^2
	-{\rm div}_y\brkts{{\rm u}_0^2\nabla_x{\rm u}^3_0}
	&\quad\textrm{in } Y^s\,,
\\
-{\rm div}_y\brkts{\kappa(y)\nabla_y{\rm u}_1^3}
	 = {\rm div}_x\brkts{\kappa(y)\nabla_y {\rm u}_0^3}&
\\\qquad\qquad
	+{\rm div}_y\brkts{\kappa(y)\nabla_x{\rm u}_0^3}
	&\quad\textrm{in } Y\,.
\end{cases}
}{1s-1}
The system \reff{1s-1} is a linear,  elliptic second order partial differential equation.
Hence solvability of \reff{1s-1} follows immediately with Lax-Milgram's Theorem by starting with 
problem \reff{1s-1}$_3$.  The fact that ${\bf u}_0$ is independent of ${\bf y}$ together with the linearity 
of \reff{1s-1} and that $\mathcal{S}_0$ only contains derivatives in ${\bf y}$, motivates to make the following Ansatz for 
${\bf u}_1(t,{\bf x},{\bf y})$, that means,
\bsplitl{
{\rm u}_1^r(t,x,y)
	= -\sum_{j=1}^N\xi^{r_j}(t,x,y)\partial_{x_j}{\rm u}^3_0(t,x)
	& \quad\textrm{for }r=1,2\,,
\\
{\rm u}_1^3(t,x,y)
	= -\sum_{j=1}^N\xi^{3_j}(y)\partial_{x_j}{\rm u}^3_0(t,x)\,.
	&
}{Au1}
We use now \reff{Au1} and the independence of ${\bf u}_0$ of ${\bf y}$ in order to rewrite 
\reff{1s-1} as a problem for $\xi^{r_j}$ for $r=1,2$ and $1\leq j\leq N$ as follows,
\bsplitl{
\begin{cases}
	-\Delta_y\xi^{r_j}(t,x,y)\partial_{x_j}{\rm u}_0^3
	-\sum_{i=1}^N\partial_{y_i}\brkts{\delta_{ij}(\partial_{x_j}{\rm u}_0^r+z_r{\rm u}_0^r\partial_{x_j}{\rm u}_0^3)}
\\\qquad\qquad\qquad
	= - {\rm div}_y\brkts{z_r{\rm u}_0^r\nabla_y\xi^{3_j}(y)}\partial_{x_j}{\rm u}_0^3
	&\textrm{in }\Omega^s\times Y^s\,,
\\
	-{\rm div}_y\brkts{
		\kappa(y)\nabla_y\xi^{3_j}(y)
	}
	+\sum_{i=1}^N\partial_{y_i}(\kappa(y)\delta_{ij})
	=0
	&\textrm{in }\Omega\times Y\,.
\end{cases}
}{xi}
We point out that under local thermodynamic equilibrium, that means, in each reference cell we have 
due to the induced separation of scales by the limit $\epsilon\to 0$,
\bsplit{
\partial_{x_j} {\rm u}_0^r
	= -z_r{\rm u}_0^r\partial_{x_j}{\rm u}_0^3\,,
}
see Definition \ref{def:LoEq}. Hence, the term with the summation 
on the left-hand side in \reff{xi}$_1$ disappears.
System \reff{xi} defines the reference cell problems for the porous media corrector functions 
$\xi^{r_j}$ for $r=1,2$ and $1\leq j\leq N$. Such correctors finally define the effective 
tensors \reff{Te}.

\medskip

\begin{rem}\label{rem:Au1}
We point out that the Ansatz \reff{Au1} is an extension from linear homogenization 
theory and is canonically chosen to account for the problem's coupled and nonlinear structure. The 
interpretation of the Ansatz \reff{Au1}$_1$ is that oscillations in the microscopic variable 
of the electrostatic potential dominate the oscillations of the concentration variables. 
\end{rem}

\medskip

\begin{lem}\label{lem:xi^3}
Let ${\rm u}_0^r\in V(\Omega_T)$. There exists a unique solution $\xi^{3_j}(y)\in W_\sharp(Y)$ for each $1\leq j\leq N$ of problem \reff{xi}$_2$, i.e., 
\bsplitl{
\begin{cases}
\textrm{Find $\xi^{3_j}\in W_\sharp(Y)$ such that} &
\\
{\rm a}^3_1\brkts{\xi^{3_j},w}
	= {\rm F}^3_1(w):=\sum_{i=1}^N\brkts{\kappa(y)\delta_{ij},\partial_{y_i}w}_Y
	&\forall w\in W_\sharp(Y)\,,
\end{cases}
}{xi3} 
where ${\rm a}^3_1\brkts{\xi^{3_j},w}:=\brkts{\kappa(y)\nabla_y\xi^{3_j},\nabla_yw}_Y$.
With $\xi^{3_l}(y)\in W_\sharp(Y)$ for $1\leq l\leq N$ also the existence and uniqueness of a solution $\xi^{r_j}(t,x,y)\in V(\Omega_T,W_\sharp(Y^s))$ of 
problem \reff{xi}$_1$ follows, that means,
\bsplitl{
\begin{cases}
\textrm{For each $1\leq j\leq N$ find $\xi^{r_j}(t,x,y)\in V(\Omega_T,W_\sharp(Y^s))$ such that}&
\\\qquad
{\rm a}^r_1\brkts{\xi^{r_j},w^r}
	= {\rm F}^{r_j}_1(w^r)
	\qquad\forall w^r\in L^2(\Omega;W_\sharp(Y^s))\,,&
\end{cases}
}{xir}
where for all $w^r=\phi^r(y)\psi^r(x)\in H^1_0(\Omega_T,W_\sharp(Y))$ we define,
\bsplitl{
{\rm a}^r_1\brkts{\xi^{r_j},w^r}
	& := \brkts{\brkts{\nabla_y\xi^{r_j}(t,x,y),\nabla_y\phi^r}_{Y^s},\psi^r}_\Omega\,,
\\
{\rm F}^{r_j}_1(w^r)
	& 
	:= 
	-\brkts{
		\brkts{ z_r{\rm u}_0^r\nabla_{y}\xi^{3_j}(y),\nabla_{y}\phi^r}_{Y^s}
		,\psi^r
	}_\Omega\,.
}{ar1Fr1}
\end{lem}

\medskip

\begin{proof}
The lemma is a consequence of Lax-Milgram's thoerem.\\
{\bf Step 1: Problem \reff{xi3}:} The assumptions of Lax-MilgramÕs are easily verified since 
$\kappa(y)$ is a strongly elliptic matrix, see Definition \ref{def:Mab}.
\\
{\bf Step 2: Problem \reff{xir}:} \emph{a)  Continuity:} For ${\rm v}^r,\, {\rm w}^r
=\phi^r(y)\psi^r(x)\in H^1_0(\Omega,W_\sharp(Y^s))$ we can estimate the bilinear form ${\rm a}_1^r$ by,
\bsplitl{
\av{{\rm a}_1^r\brkts{{\rm v}^r,{\rm w}^r}}
	= \av{\brkts{\brkts{
	\nabla_y{\rm v}^r,\nabla_y\phi^r
	}_{Y^s}
	,\psi^r}_\Omega}
	\leq 
	\N{{\rm v}^r}{L^2(\Omega;W_\sharp(Y^s))}\N{\phi^r}{W_\sharp(Y^s)}\N{\psi^r}{L^2(\Omega)}\,,
}{Esar1}
and hence continuity follows.
\\
\emph{b)  Coercivity:} For ${\rm v}^r,\, {\rm w}^r
=\phi^r(y)\psi^r(x)\in H^1_0(\Omega,W_\sharp(Y^s))$ we derive a lower bound by,
\bsplitl{
{\rm a}_1^r\brkts{{\rm v}^r,{\rm v}^r}
	= \brkts{\brkts{
	\nabla_y{\rm v}^r,\nabla_y{\rm v}^r
	}_{Y^s}
	,1}_\Omega
	=
	\N{\nabla_y{\rm v}^r}{L^2(\Omega;L^2(Y^s))}^2\,.
}{Esar2}
\\
\emph{c)  ${\rm F}^{r_j}_1$ is linear and continuous:} For ${\rm w}^r
=\phi^r(y)\psi^r(x)\in H^1_0(\Omega,W_\sharp(Y^s))$ we estimate ${\rm F}^{r_j}_1$ 
by,
\bsplitl{
\av{{\rm F}^{r_j}_1(w^r)}
	& \leq 
	\av{\brkts{
		\brkts{z_r{\rm u}_0^r\nabla_y\xi^{3_j},\nabla_y\phi^r
		}_{Y^s},\psi^r
	}_\Omega}
\\&
	\leq 
	C\N{\nabla_y\xi^{3_j}}{L^2(Y^s)}\N{\nabla_y\phi^r}{L^2(Y^s)}\N{{\rm u}_0^r}{L^2(\Omega)}
		\N{\psi^r}{L^2(\Omega)}\,,
}{EsFrj}
and continuity follows with $\xi^{3_j}\in W_\sharp(Y)$ obtained in Step 1.
\end{proof}

\medskip

{\bf (3) Problem for terms of order $\mathcal{O}(1)$:}
We collect all the terms which do not contain any factor $\epsilon$. As a result we end up with the following 
equation
\bsplitl{
\mathcal{O}(1):\qquad
\begin{cases}
-\mathcal{S}_0({\bf u}_0,{\bf y}) {\bf u}_2
	= 
	\tilde{\mathcal{S}}_1({\bf u}_0,{\bf u}_1,{\bf y}){\bf u}_1
	&
\\\qquad\qquad
	+\tilde{\mathcal{S}}_2({\bf u}_0,{\bf u}_1,{\bf u}_2,{\bf y}){\bf u}_0
	+{\bf I}({\bf u}_0)
	+{\bf D}_{\bf t}{\bf u}_0&
\\
{\bf u}_2(\cdot,{\bf y})
	\textrm{ $Y$-periodic\,,}
\end{cases}
}{s-0}
where $\tilde{\cg S}_1:={\cg S}_1({\bf u}_0,{\bf y})+\ebrkts{{\cg S}_0({\bf u}_1,{\bf y})}^N$ and 
$\tilde{\cg S}_2 := {\cg S}_2({\bf u}_0,{\bf y})+\ebrkts{{\cg S}_1({\bf u}_1,{\bf y})}^N
+\ebrkts{{\cg S}_0({\bf u}_1,{\bf y})}^N$.
In order to study solvability of \reff{s-0} it is an advantage to write \reff{s-0} explicitly for 
each physical quantity. For $r=1,2$ we have
\bsplitl{
\mathcal{O}(1):
\begin{cases}
-\Delta_y {\rm u}_2^r
	= {\rm div}_y\brkts{z_r{\rm u}_0^r\nabla_y {\rm u}^3_2}&
\\\qquad\qquad
	+\bigl\{
		{\rm div}_x\brkts{z_r{\rm u}_0^r\nabla_y{\rm u}_1^3}
		+{\rm div}_y\brkts{z_r{\rm u}_0^r\nabla_x{\rm u}^3_1}&
\\\qquad\qquad
		+{\rm div}_x\nabla_y{\rm u}^r_1
		+{\rm div}_y\nabla_x{\rm u}^r_1
		+{\rm div}_y\brkts{z_r{\rm u}_1^r\nabla_y{\rm u}^3_1}
	\bigr\}&
\\\qquad\qquad
	+\bigl\{
		\Delta_x {\rm u}_0^r
		+{\rm div}_x\brkts{z_r{\rm u}_0^r\nabla_x{\rm u}_0^3}
		+{\rm div}_x\brkts{z_r{\rm u}_1^r\nabla_y{\rm u}_0^3}&
\\\qquad\qquad
		+{\rm div}_y\brkts{z_r{\rm u}_1^r\nabla_x {\rm u}_0^3}
		+{\rm div}_y\brkts{z_r{\rm u}^r_2\nabla_y{\rm u}_0^3}
	\bigr\}
	 -\partial_t{\rm u}_0^r
	 &\textrm{in }\Omega^s\times Y^s\,,
\\
-{\rm div}_y\brkts{\kappa(y)\nabla_y{\rm u}_2^3}
	= {\rm div}_x\brkts{\kappa(y)\nabla_y{\rm u}_1^3}
	+{\rm div}_y\brkts{\kappa(y)\nabla_x{\rm u}_1^3}&
\\\qquad\qquad
	+{\rm div}_x\brkts{\kappa(y)\nabla_x{\rm u}_0^3}
	+({\rm u}_0^1-{\rm u}_0^2)\chi_{Y^s}
	&\textrm{in }Y\,.
\end{cases}
}{1s-0}
Since a solvability constraint implies the effective equation for the macroscopic quantities, we first 
achieve the well-posedness for the system \reff{1s-0}.

\medskip

\begin{lem}\label{lem:ur2u32}
The problem \reff{1s-0}$_2$, i.e.,
\bsplitl{
\begin{cases}
\textrm{Find ${\rm u}_2^3(\cdot,\cdot,y)\in W_\sharp(Y)$ such that}
&\\
{\rm a}_2^3\brkts{{\rm u}_2^3,w^3}
	= {\rm F}_2^3(w^3)
	&\forall w^3\in W_\sharp(Y)\,,
\end{cases}
}{u32}
has a unique solution ${\rm u}_2^3\in H^1(\Omega_T,W_\sharp(Y))$ where we define,
\bsplitl{
{\rm a}_2^3\brkts{{\rm u}_2^3,w^3}
	& := \brkts{\kappa(y)\nabla_y{\rm u}_2^3,\nabla_y w^3}_Y\,,
\\
{\rm F}_2^3(w^3)
	& := \brkts{{\rm div}_x\brkts{\kappa(y)\nabla_y {\rm u}_1^3},w^3}_Y
	-\brkts{\kappa(y)\nabla_x{\rm u}_1^3,\nabla_yw^3}_Y
\\&\quad
	+\brkts{{\rm div}_x\brkts{\kappa(y)\nabla_x{\rm u}_0^3},w^3}_Y
	+\brkts{{\rm u}_0^1-{\rm u}_0^2,w^3}_{Y^s}\,.
}{a32F32}
For $r=1,2$ equation \reff{1s-0}$_1$, that means the following problem,
\bsplitl{
\begin{cases}
\textrm{Find ${\rm u}_2^r(\cdot,\cdot,y)\in W_\sharp(Y^s)$ such that}&
\\
{\rm a}_2^r\brkts{{\rm u}_2^r,w^r}
	= {\rm F}_2^r(w^r)
	&\forall w^r\in W_\sharp(Y^s)\,,
\end{cases}
}{ur2}
has a unique solution ${\rm u}_2^r(\cdot,\cdot,y)\in W_\sharp(Y^s)$, where we for $z_1=1,\, z_2=-1$ 
define, 
\bsplitl{
{\rm a}_2^r\brkts{{\rm u}_2^r,w^r}
	& := \brkts{\nabla_y{\rm u}_2^r,\nabla_y w^r}_{Y^s}
\\
{\rm F}_2^r(w^r)
	& := - \brkts{z_r{\rm u}_0^r\nabla_y{\rm u}_2^3,\nabla_yw^r}_{Y^s}
	+\brkts{{\rm div}_x\brkts{z_r{\rm u}_0^r\nabla_y{\rm u}_1^3},w^r}_{Y^s}
\\&\quad
	-\brkts{z_r{\rm u}_0^r\nabla_x{\rm u}_1^3,\nabla_yw^r}_{Y^s}
	+\brkts{{\rm div}_x\nabla_y{\rm u}_1^r,w^r}_{Y^s}
		-\brkts{\partial_t{\rm u}_0^r,w^r}_{Y^s}
\\&\quad
	-\brkts{\nabla_x{\rm u}_1^r,\nabla_y w^r}_{Y^s}
	-\brkts{z_r{\rm u}_1^r\nabla_y{\rm u}_1^3,\nabla_yw^r}_{Y^s}
	+\brkts{\Delta_x{\rm u}_0^r,w^r}_{Y^s}
\\&\quad
	+\brkts{{\rm div}_x\brkts{z_r{\rm u}_0^r\nabla_x{\rm u}_0^r},w^r}_{Y^s}
	-\brkts{z_r{\rm u}_1^r\nabla_x{\rm u}_0^3,\nabla_yw^r}_{Y^s}
	\,.
}{ar2Fr2}
\end{lem}

\medskip

\begin{proof}
{\bf Step 1:} Theorem \ref{thm:ElEx} and Lemma \ref{lem:WpDual} immediately provide existence and uniqueness of a solution 
${\rm u}_2^3\in W_\sharp(Y)$, if we have,
\bsplitl{
\abrkts{{\rm F}^3_2,1}_{(H^1_\sharp(Y^s))',H^1_\sharp(Y^s)}
	= 0\,.
}{u32SoCo}
Equation \reff{u32SoCo} reads in terms of physical quantities as,
\bsplitl{
\brkts{{\rm div}_x\brkts{\kappa(y)\nabla_y {\rm u}_1^3},1}_Y
	-\brkts{{\rm div}_x\brkts{\kappa(y)\nabla_x{\rm u}^3_0},1}_Y
	+\brkts{{\rm u}_0^1-{\rm u}_0^2}_{Y^s}
	= 0\,.
}{u32SoCo1}
This equation defines the upscaled formulation \reff{UpScRe}$_2$ for the electric potential such 
that \reff{u32SoCo} holds true.
\\
{\bf Step 2:} Again, solvability follows by Theorem \ref{thm:ElEx} after verification of ${\rm F}_2^r\in (W_\sharp(Y^s))'$. Due to 
Theorem \ref{thm:ElEx} and Lemma \ref{lem:WpDual}, it must hold for $r=1,2$ that
\bsplitl{
\abrkts{{\rm F}^r_2,1}_{(H^1_\sharp(Y^s))',H^1_\sharp(Y^s)}
	= 0\,,
}{ur2SoCo}
which reads in explicit form as follows,
\bsplitl{
\brkts{{\rm div}_x\nabla_y{\rm u}_1^r,1}_{Y^s}
	& +\brkts{\Delta_x{\rm u}_0^r,1}_{Y^s}
	-\brkts{\partial_t{\rm u}_0^r,1}_{Y^s}
\\
	&+\brkts{{\rm div}_x\brkts{z_r{\rm u}_0^r\nabla_y{\rm u}_1^3},1}_{Y^s}
	+\brkts{{\rm div}_x\brkts{z_r{\rm u}_0^r\nabla_y{\rm u}_0^3},1}_{Y^s}
	=0\,.
}{SoCo}
Since we use \reff{SoCo} represents the by Lemma \ref{lem:ExReMo} well-posed effective model, we herewith guarantee that \reff{ur2SoCo} holds.
\end{proof}

\medskip

With representation \reff{Au1}$_2$ we can rewrite \reff{u32SoCo1} in the following way,
\bsplitl{
-\sum^N_{i,j,k=1}\brkts{
		\partial_{x_i}\brkts{
			\kappa(y)\brkts{
				\delta_{ik}-\delta_{ij}\partial_{y_j}\xi^{3_k}
			}\partial_{x_k}{\rm u}_0^3
		},1
	}_Y
	=\av{Y^s}({\rm u}_0^1 -{\rm u}_0^2)\,.
}{Effu3}
Using \reff{u03} allows us to write \reff{Effu3} more intuitively by,
\bsplitl{
\av{Y}\sum^N_{i,j,k=1}
	\epsilon^0_{ik}
	\frac{\partial^2 {\rm u}_0^3}{\partial x_i\partial x_k}
	=\av{Y^s}({\rm u}_0^1 -{\rm u}_0^2)\,,
}{Effu3-1}
such that $\pmb{\epsilon}^0:=\brcs{\epsilon^0_{ik}}_{1\leq i,k\leq N}$ is defined by
\bsplitl{
\epsilon^0_{ik}
	:=
	-\frac{1}{\av{Y}}
	\brkts{
		\brkts{
			\kappa(y)\brkts{
				\delta_{ik}-\delta_{ij}\partial_{y_j}\xi^{3_k}
			}
		},1
	}_Y\,.
}{eps-ik}
We can write down in the same way effective equations for equations \reff{ur2} by using 
\reff{Au1}$_1$. That means we rewrite \reff{SoCo} in the following way,
\bsplitl{
&\brkts{\partial_{x_i}(\delta_{ik}\partial_{x_k}{\rm u}_0^r),1}_{Y^s}
+\brkts{\partial_{x_i}\brkts{\brcs{-\delta_{ij}\partial_{y_j}\xi^{r_k}}\partial_{x_k}{\rm u}_0^3},1}_{Y^s}
	-\brkts{\partial_t {\rm u}_0^r,1}_{Y^s}
\\
&	+\brkts{\partial_{x_i}\brkts{z_r{\rm u}_0^r\brcs{\delta_{ik}-\delta_{ij}\partial_{y_j}\xi^{3_k}}\partial_{x_k}{\rm u}_0^3},1}_{Y^s}
	= 0\,.
}{Effur-1}
If we apply \reff{u01u02} we can even further simplify \reff{Effur-1} to 
\bsplitl{
	\av{Y^s}\partial_t {\rm u}_0^r
	-\av{Y^s}\Delta{\rm u}_0^r
	& + \partial_{x_i}\brkts{
		\int_{Y^s}\brcs{\delta_{ij}\partial_{y_j}\xi^{r_k}(x,y)}\,dy\partial_{x_k}{\rm u}_0^3
	}
\\
	& - \partial_{x_i}\brkts{
		z_r{\rm u}_0^r\int_{Y^s}\brcs{\delta_{ik}-\delta_{ij}\partial_{y_j}\xi^{3_k}(y)}\,dy\partial_{x_k}{\rm u}_0^3
	}
	= 0\,.
}{Effur-2}
As a consequence of \reff{Effur-2} we obtain the effective diffusion-related tensor $\mathbb{D}^r:=\brcs{{\rm D}_{ik}^r}_{1\leq i,k\leq N}$, i.e.,
\bsplitl{
{\rm D}^r_{ik}
	:= \frac{1}{\av{Y}}
	\int_{Y^s}\brcs{\delta_{ij}\partial_{y_j}\xi^{r_k}(t,x,y)}\,dy
	\quad\forall i,k=1,\dots,N\,,
}{Dkl}
and the electro-diffusion-related tensor $\mathbb{M}:=\brcs{{\rm M}_{ik}}_{1\leq i,k\leq N}$, i.e., 
\bsplitl{
{\rm M}_{ik}
	:= \frac{1}{\av{Y}}
	\int_{Y^s}\brcs{\delta_{ik}-\delta_{ij}\partial_{y_j}\xi^{3_k}(y)}\,dy
	\quad\forall i,k=1,\dots,N\,.
}{Mik}
The definitions \reff{Dkl} and \reff{Mik} finally provide the porous media 
approximation of the Nernst-Planck-Poisson equations,
\bsplitl{
{\bf upscaled\,\,\,\, model:}\quad
\begin{cases}
p\partial_t {\rm u}_0^r
	-p \Delta {\rm u}_0^r
	+ {\rm div}\brkts{\mathbb{D}^r(t,x)\nabla {\rm u}_0^3}
	-{\rm div}\brkts{z_r{\rm u}_0^r\mathbb{M}\nabla{\rm u}_0^3}
	= 0\,,&
\\
-{\rm div}\brkts{\pmb{\epsilon}^0\nabla {\rm u}_0^3}
	= p({\rm u}_0^1 - {\rm u}_0^2)\,,&
\end{cases}
}{pmPNP}
where $p:=\frac{\av{Y^s}}{\av{Y}}$ is the porosity, $r=1,2$, $z_1=1$, and $z_2=-1$.

\medskip

{\bf (4) Derivation of the second order correctors:} 
In order to compute the second order corrector for ${\rm u}_0^3$ we use 
\reff{Au1}$_2$ in equation \reff{1s-0}$_2$, i.e., 
\bsplitl{
-{\rm div}_y\brkts{\kappa(y)\nabla_y{\rm u}_2^3}
	&= 
	-\sum_{i,j,k=1}^N\partial_{x_i}\brkts{\kappa(y)\delta_{ij}\partial_{y_j}\xi^{3_k}\partial_{x_k}{\rm u}_0^3}
\\&	
	-\sum_{i,j,k=1}^N\partial_{y_i}\brkts{\kappa(y)\delta_{ij}\partial_{x_j}(\xi^{3_k}\partial_{x_k}{\rm u}_0^3}
\\
	&+{\rm div}_x\brkts{\kappa(y)\nabla_x{\rm u}_0^3}
	+({\rm u}_0^1-{\rm u}_0^2)\chi_{Y^s}
	\qquad\textrm{in }Y\,.
}{u23_1}
Inserting  equation \reff{Effu3-1} into \reff{u23_1} leads to the following problem,
\bsplitl{
-{\rm div}_y\brkts{\kappa(y)\nabla_y{\rm u}_2^3}
	&=
	-\sum_{k,l=1}^N\epsilon^0_{kl}\frac{\partial^2 {\rm u}_0^3}{\partial x_k\partial x_l}
	-\sum_{i,j,k=1}^N\kappa(y)\delta_{kj}\partial_{y_j}\xi^{3_l}\frac{\partial^2{\rm u}_0^3}{\partial x_l\partial x_k}
\\&	
	-\sum_{i,j,k=1}^N\partial_{y_i}\brkts{\kappa(y)\delta_{ij}\xi^{3_k}}
		\frac{\partial^2 {\rm u}_0^3}{\partial x_l\partial x_k}
	+\sum_{j,l=1}^N\kappa(y)\delta_{jl}\frac{\partial^2 {\rm u}_0^3}{\partial x_j\partial x_l}
	\qquad\textrm{in }Y\,.
}{u23_2}
With equation \reff{u23_2} the right-hand side ${\rm F}^3_2$ in \reff{a32F32}$_2$ can be rewritten 
by,
\bsplitl{
\abrkts{{\rm F}^3_2,w^3}_{(W_\sharp(Y))',W_\sharp(Y)}
	& = \sum^N_{k,l=1} \Biggl[
		-\epsilon^0_{kl} \int_Yw^3\,dy
		-\sum_{i,j=1}^N \int_Y \partial_{y_i}\brkts{\kappa (y)\delta_{ij}\delta_{kj}\xi^{3_l}}w^3\,dy
\\&
		-\sum_{j,k,l=1}^N\int_Y\kappa(y)\delta_{kj}\partial_{y_j}\brkts{\xi^{3_l}-y_l}w^3\,dy
	\Biggr]
	\frac{\partial^2 {\rm u}_0^3}{\partial x_k\partial x_l}
	\quad\forall w^3\in W_\sharp(Y)\,.
}{F32_1}
The same arguments as those for \reff{Au1} suggest to look for a function ${\rm u}^3_2$ 
of the following form,
\bsplitl{
{\rm u}_2^3(t,x,y)
	:= \sum_{k,l=1}^N\zeta^{3_{kl}}(y)\frac{\partial^2 {\rm u}_0^3}{\partial x_k\partial x_l}
	\,,
}{zeta}
where $\zeta^{3_{kl}}$ is the solution of,
\bsplitl{
-{\rm div}_y\brkts{\kappa(y)\nabla_y\zeta^{3_{kl}}}
	&=
	-\epsilon^0_{kl}
	-\sum_{i,j=1}^N \partial_{y_i}\brkts{\kappa (y)\delta_{ij}\delta_{kj}\xi^{3_l}}
		-\sum_{j=1}^N\kappa(y)\delta_{kj}\partial_{y_j}\brkts{\xi^{3_l}-y_l}
	\quad\textrm{in }Y\,.
}{zeta1}

\medskip

\begin{lem}\label{lem:zeta3}
There exists a unique $\zeta^{3_{kl}}$ for $1\leq k,l\leq N$ that solves equation 
\reff{zeta1} in the weak sense, that means, $\zeta^{3_{kl}}$ is a unique solution of the 
following problem,
\bsplitl{
\begin{cases}
\textrm{Find $\zeta^{3_{kl}}\in W_\sharp(Y)$ such that}
\\
{\rm a}_2^3(\zeta^{3_{kl}},w^3)
	= {\rm F}_2^{3_{kl}}(w^3)
	& \forall w^3\in W_\sharp(Y)\,.
\end{cases}
}{ExUnZeta3}
where we define,
\bsplitl{
{\rm a}_2^3(\zeta^{3_{kl}},w^3)
	& := \brkts{\kappa(y)\nabla_y\zeta^{3_{kl}},\nabla_y w^3}_Y\,,
\\
{\rm F}_2^{3_{kl}}(w^3)
	& := -\brkts{\epsilon^0_{kl},w^3}_Y
	-\sum_{i,j=1}^N\brkts{\partial_{y_i}\brkts{\kappa(y)\delta_{ij}\delta_{kj}\xi^{3_l}},w^3}_Y
\\&\quad
	-\sum_{j=1}^N\brkts{\kappa(y)\delta_{kj}\partial_{y_j}\brkts{\xi^{3_l}-y_l},w^3}_Y\,.	
}{a23F3kl2}
\end{lem}

\medskip

\begin{proof}
The existence and uniqueness is an immediate consequence of Theorem \reff{thm:ElEx} and 
Lemma \reff{lem:WpDual}.
\end{proof}

\medskip

The same considerations as those for the derivation of \reff{zeta1} can be applied to 
\reff{ar2Fr2} (or to \reff{1s-0} in the context of its classical formulation). 
With \reff{zeta} and \reff{Au1} we can rewrite \reff{1s-0} as,
\bsplitl{
-\Delta_y {\rm u}_2^r
	& = \frac{1}{p}\sum_{k,l=1}^N {\rm D}^r_{kl} \frac{\partial^2 {\rm u}_0^3}{\partial x_k\partial x_l}
	-\frac{1}{p}\sum_{k,l=1}^N \partial_{x_k}\brkts{z_r{\rm u}_0^r {\rm M}_{kl}\partial_{x_l} {\rm u}_0^3}
	-\sum_{k,l=1}^N\frac{\partial^2 {\rm u}_0^r}{\partial x_r\partial x_l}
\\&
	-\sum_{j,k,l=1}^N \delta_{kj}\partial_{y_j}\xi^{r_l}\frac{\partial^2 {\rm u}_0^3}{\partial x_l\partial x_k}
	-\sum_{i,j,k=1}^N \partial_{y_i}\brkts{\delta_{ij}\xi^{r_k}}\frac{\partial^2 {\rm u}_0^3}{\partial x_j\partial x_k}
	+\sum_{j,l=1}^N \delta_{jl}\frac{\partial^2 {\rm u}_0^r}{\partial x_j\partial x_l}
\\&
	-\sum_{j,k,l=1}^N \brcs{
		z_r\delta_{kj}\partial_{y_j}\xi^{3_l}\frac{\partial {\rm u}_0^r}{\partial x_l}
			\frac{\partial {\rm u}_0^3}{\partial x_k}
		+z_r{\rm u}_0^r\delta_{kj}\partial_{y_j}\xi^{3_l}\frac{\partial^2 {\rm u}_0^3}{\partial x_l\partial x_k}
	}
\\&
	+\sum_{i,j,k,l=1}^N \partial_{y_i}\brkts{\delta_{ij}\partial_{y_j}\zeta^{3_{kl}}}
		z_r{\rm u}_0^r\frac{\partial^2 {\rm u}_0^3}{\partial x_k\partial x_l}
	-\sum_{i,j,k=1}^N \partial_{y_i}\brkts{\delta_{ij}\xi^{r_k}}z_r
		\frac{\partial {\rm u}_0^r}{\partial x_k}\frac{\partial {\rm u}_0^3}{\partial x_j}
\\&
	-\sum_{i,j,k=1}^N z_r{\rm u}_0^r\partial_{y_i}\brkts{\delta_{ij}\xi^{3_k}}
		\frac{\partial^2 {\rm u}_0^3}{\partial x_k\partial x_j}
	+\sum_{i,j,k,l=1}^N z_r\partial_{y_i}\brkts{\delta_{ij}\xi^{r_k}\partial_{y_j}\xi^{3_l}}
		\frac{\partial {\rm u}_0^r}{\partial x_k}\frac{\partial {\rm u}_0^3}{\partial x_l}
\\&
	+\sum_{j,l=1}^N \brcs{z_r\delta_{jl}\frac{\partial {\rm u}_0^r}{\partial x_j}
		\frac{\partial {\rm u}_0^3}{\partial x_l}
		+ z_r {\rm u}_0^r\delta_{jl}\frac{\partial^2 {\rm u}_0^3}{\partial x_j\partial x_l}
	}
	\qquad\textrm{in }Y^s\,.
}{u2r_1}
We make now a corresponding Ansatz for the functions ${\rm u}_2^r$ with the indices $r=1,2$ under the 
same considerations as those for \reff{zeta}, i.e.,
\bsplitl{
{\rm u}_2^r(t,x,y)
	= \sum_{k,l=1}^N \zeta^{r_{kl}}(t,x,y){\rm u}^r_0\,.
}{zetar}
After inserting definition \reff{zetar} into  \reff{u2r_1} we obtain an equation 
for the second order corrector functions $\zeta^{r_{kl}}$, that means,
\bsplitl{
-\sum_{k,l=1}^N 
	& \Delta_y\zeta^{r_{kl}}{\rm u}_0^r
	 = \sum_{k,l=1}^N 
	\Biggl[
		\frac{1}{p}{\rm D}^r_{kl}
		-\sum_{j=1}^N \delta_{kj}\partial_{y_j}\brkts{\xi^{r_l}}
		-\sum_{i,j=1}^N \partial_{y_i}\brkts{\delta_{ij}\delta_{kj}\xi^{r_k}}
	\Biggr]
		\frac{\partial^2 {\rm u}_0^3}{\partial x_k\partial x_l}
\\&
	+\sum_{k,l=1}^N 
	\Biggl[
		-\frac{z_r}{p}{\rm M}_{kl}
		-\sum_{j=1}^N z_r\delta_{kj}\partial_{y_j}\brkts{\xi^{3_l}-y_l}
\\&
		+\sum_{i,j=1}^N z_r\partial_{y_i}\brkts{\delta_{ij}\delta_{kj}\xi^{r_k}
			\brkts{\partial_{y_j}\xi^{3_l}-1} }
	\Biggr]
		\frac{\partial {\rm u}_0^r}{\partial x_k}\frac{\partial {\rm u}_0^3}{\partial x_l}
\\&
	+ \sum_{k,l=1}^N	\Biggl[
		-\sum_{j=1}^N  \delta_{kj}\partial_{y_j}\brkts{\xi^{3_l}-y_l}
\\&
		-\sum_{i,j=1}^N \partial_{y_i} \brkts{
				\delta_{ij}\delta_{kj}\brkts{
					\xi^{3_k} -\partial_{y_j}\zeta^{3_{kl}}
				}
			}
	\Biggr]
		z_r{\rm u}_0^r \frac{\partial^2 {\rm u}_0^3}{\partial x_k\partial x_l}
	+\sum_{k,l,j=1}^N\brkts{
		\frac{\partial y_l}{\partial y_j}- \delta_{kl}
	}\frac{\partial^2 {\rm u}_0^r}{\partial x_k\partial x_l}	
\,.
}{zetar1}
In order to guarantee solvability of equation \reff{zetar1}, we assume that 
${\rm u}_0^r>0$ for ${\rm u}_0^r(0,x)\geq \eta >0$. This can be obtained by special test function 
techniques as applied in \cite{Schmuck2009} to prove non-negativity of solutions. Existence and uniqueness 
of the corrector functions defined by equations 
\reff{zeta1} and \reff{zetar1} is achieved in the following 

\begin{lem}\label{lem:zetar}
Let ${\rm u}_0^r\in V(\Omega_T)$, ${\rm u}_0^r(t,x)>0$ for all $(t,x)\in\Omega_T$ and $1\leq r\leq N$, and assume that the reference cells $Y$ are in local thermodynamic 
equilibrium, see Definition \reff{def:LoEq}. Then, there exists a unique $\zeta^{r_{kl}}$ for $1\leq k,l\leq N$ that solves equation \reff{zeta1} in the weak sense. That means, $\zeta^{r_{kl}}$ is a unique solution of the 
following problem,
\bsplitl{
\begin{cases}
\textrm{Find $\zeta^{r_{kl}}(t,x,y)\in V(\Omega_T,W_\sharp(Y^s))$ for $1\leq k,l\leq N$ such that}
\\
{\rm a}_2^r(\zeta^{r_{kl}},w^r)
	= {\rm F}_2^{r_{kl}}(w^r)
	\qquad\qquad\qquad \forall w^r\in H^1_0(\Omega,W_\sharp(Y^s))\,.&
\end{cases}
}{ExUnZetar}
where we define for $w^r:=\phi^r(x)\psi^r(y)$,
\bsplitl{
{\rm a}_2^r(\zeta^{r_{kl}},w^r)
	& := \brkts{\brkts{{\rm u}_0^r\kappa(y)\nabla_y\zeta^{r_{kl}},\nabla_y \phi^r}_{Y^s},\psi^r}_{\Omega}\,,
\\
{\rm F}_2^{r_{kl}}(w^r)
	& := \Biggl(
		\brkts{
		 {\rm D}^r_{kl}
		-\sum_{j=1}^N \delta_{kj}\partial_{y_j}\brkts{\xi^{r_l}-y_l}
		-\sum_{i,j=1}^N \partial_{y_i}\brkts{\delta_{ij}\delta_{kj}\xi^{r_k}}
		,\phi^r}_{Y^s}
		\frac{\partial^2 {\rm u}_0^3}{\partial x_k\partial x_l} 
\\&
	+
		\Biggl(
		-{\rm M}_{kl}
		-\sum_{j=1}^N \delta_{kj}\partial_{y_j}\brkts{\xi^{3_l}-y_l}
\\&
		+\sum_{i,j=1}^N \partial_{y_i}\brkts{\delta_{ij}\delta_{kj}\xi^{r_k}
			\brkts{\partial_{y_j}\xi^{3_l}-1} }
		,\phi^r\Biggr)_{Y^s}
		z_r \brkts{\frac{\partial {\rm u}_0^r}{\partial x_k}\frac{ \partial {\rm u}_0^3}{\partial x_l}}
\\&
	+\sum_{j=1}^N\brkts{\frac{\partial y_l}{\partial y_j}-\delta_{kl}}_{Y^s}
		\frac{\partial^2 {\rm u}_0^r}{\partial x_k\partial x_l} 
\\&
	+ 
		\Biggl(
		-{\rm M}_{kl}
		-\sum_{j=1}^N  \delta_{kj}\partial_{y_j}\brkts{\xi^{3_l}-y_l}
\\&
		-\sum_{i,j=1}^N \partial_{y_i} \brkts{
				\delta_{ij}\delta_{kj}\brkts{
					\xi^{3_k} -\partial_{y_j}\zeta^{3_{kl}}
				}
			}
		,\phi^r\Biggr)_{Y^s}
		z_r{\rm u}_0^r\frac{\partial^2 {\rm u}_0^3}{\partial x_k\partial x_l} 
	,\psi^r\Biggr)_\Omega
	\,.	
}{a23Frkl2}
\end{lem}

\begin{proof}
We apply Lax-Milgram's theorem in three steps:\\
{\bf a) Continuity:} For $v(x,y),w(x,y)\in H^1_0(\Omega;W_\sharp(Y^s))$ we have
\bsplitl{
\av{{\rm a}_2^{r}(v,w)} 
	=\av{\brkts{\kappa(y)\nabla_y v,\nabla_y w}_{Y^s,\Omega}}
	\leq C\N{\nabla_y v}{H^1(\Omega;L^2(Y^s))}\N{\nabla_y w}{H^1(\Omega;L^2(Y^s))}\,,
}{a)}
which proves continuity since $\N{\nabla_y v}{L^2(Y^s)}$ defines a norm on $W_\sharp(Y)$ and 
$\kappa \in M(c,C,Y)$.
\\
{\bf b) Coercivity:} For $v(x,y)\in H^1_0(\Omega;W_\sharp(Y^s))$ we can estimate
\bsplitl{
\av{{\rm a}_2^r(v,v)}
	 = \av{\brkts{{\rm u}_0^r\kappa(y)\nabla_y v,\nabla_y v}_{Y^s,\Omega} }
	 \geq c\N{\nabla_y v}{H^1(\Omega;L^2(Y^s))}^2\,.
}{b)}
\\
{\bf c) ${\rm F}^{r_{kl}}_2$ is linear and continuous:} In order to keep the clear representation we only 
study the two critical terms, i.e.,
\bsplitl{
(I) & := \brkts{ 
		\brkts{
		\delta_{ij}\delta_{kj}\xi^{r_k}(\partial_{y_j}\xi^{3_l}-1),\partial_{y_i}\phi^r
		}_{Y^s}
		\brkts{\frac{\partial {\rm u}_0^r}{\partial x_k}\frac{ \partial {\rm u}_0^3}{\partial x_l}}
		,\psi^r
	}_\Omega
\\
(II) & := \brkts{
		\brkts{
			\delta_{ij}\delta_{kj}(\xi^{3_k}-\partial_{y_j}\zeta^{3_{kl}})
			,\partial_{y_i}\phi^r
		}_{Y^s}
		z_r{\rm u}_0^r\frac{\partial^2 {\rm u}_0^3}{\partial x_k\partial x_l}
		,\psi^r
	}_\Omega\,.
}{TIandII}
We need to show that $\av{(I)},\av{(II)}\leq C\N{w}{H^1(\Omega_T;W_\sharp(Y^s))}$. The second 
term $(II)$ can be estimated for $w^r(x,y)=\phi^r(x)\psi^r(y)\in H^1_0(\Omega;W_\sharp(Y^s))$ by,
\bsplitl{
\av{(II)}
	& \leq C\brkts{\N{\xi^{3_k}}{H^1(\Omega;L^2(Y^s))} 
			+ \N{\nabla_y\zeta^{3_{kl}}}{L^2(Y^s)}
		}\N{\nabla_y\phi^r}{L^2(Y^s)}
\\&\quad
		\N{
			\frac{\partial^2 {\rm u}_0^3}{\partial x_k\partial x_l}
		}{L^2(\Omega)}
		\N{{\rm u}_0^r}{H^1(\Omega)}
		\N{\psi^r}{H^1(\Omega)}\,.
}{(II)}
With the regularity available for $\xi,\zeta$ and ${\rm u}_0^r$, the continuity of the second term $(II)$ 
follows. We estimate the first term $(I)$ by, 
\bsplitl{
\av{(I)}
	& \leq C\N{\xi^{r_k}}{H^1(\Omega;W_\sharp(Y^s))}\biggl(
		\N{\nabla_y\xi^{3_l}}{L^3(Y^s)}
\\&\quad
		+1
	\biggr)\N{\nabla_y\phi^r}{L^2(Y^s)}
		\N{
			 {\rm u}_0^3
		}{H^2(\Omega)}
		\N{{\rm u}^r_0}{H^1(\Omega)}
		\N{\psi^r}{H^1(\Omega)}\,,
}{(I)}
which implies the desired continuity.
\end{proof}

%% file: existence.tex
For each 
${\rm v}^r\in V(\Omega_T)$, $r=1,2$, consider the linear system 
\bsplitl{
\begin{cases}
\qquad
p\partial_t {\rm u}^r
	- p\Delta {\rm u}^r
	= - {\rm div}\brkts{\ol{\mathbb{D}}^{r}(t,x)\nabla{\rm v}^3}
	+ {\rm div}\brkts{z_r{\rm v}^r{\mathbb M}\nabla{\rm v}^3}
	&  \quad\textrm{in }\Omega_T\,,
\\\qquad
-{\rm div}\brkts{\pmb{\epsilon}^0\nabla {\rm v}^3}
	= p\brkts{{\rm v}^1
	-{\rm v}^2}
	&  \quad\textrm{in }\Omega_T\,,
\end{cases}
}{LAVM}
where $\ol{\mathbb{D}}^{r}(t,x)$ indicates that its $t$ and 
$x$ dependence originate from ${\rm v}^r$ for $r=1,2$.
The choice of ${\rm v}^r\in V(\Omega_T)$ guarantees that the right-hand side in \reff{LAVM}$_1$ is in 
$L^2(\Omega)$. Hence, there exists a unique solution ${\rm u}^r\in L^2(0,T;H^2_0(\Omega))$, 
$\partial_t{\rm u}^r \in L^2(0,T;L^2(\Omega))$ by standard parabolic theory.

\medskip

In the same way, take $\tilde{\rm v}^r\in V(\Omega_T)$, $r=1,2$, and let $\tilde{\rm u}^r$ solve,
\bsplitl{
\begin{cases}
\qquad
p\partial_t \tilde{\rm u}^r
	- p\Delta \tilde{\rm u}^r
	= -{\rm div}\brkts{\tilde{\mathbb{D}}^{r}(t,x)\nabla\tilde{\rm v}^3}
	+ {\rm div}\brkts{z_r\tilde{\rm v}^r{\mathbb M}\nabla\tilde{\rm v}^3}
	&  \quad\textrm{in }\Omega_T\,,
\\\qquad
-{\rm div}\brkts{\pmb{\epsilon}^0\nabla\tilde{\rm v}^3}
	= p\brkts{\tilde{\rm v}^1
	-\tilde{\rm v}^2}
	&  \quad\textrm{in }\Omega_T\,.
\end{cases}
}{LAVM2}

\medskip

After subtracting \reff{LAVM2} from \reff{LAVM}, we obtain the following 
equation for $\hat{\rm u}^r := {\rm u}^r - \tilde{\rm u}^r$, 
$\hat{\mathbb D}^{0,r} := \ol{\mathbb D}^{0,r} - \tilde{\mathbb D}^{0,r}$, 
for $r=1,2$, i.e.,
\bsplitl{
\begin{cases}
\quad
	p\partial_t \hat{\rm u}^r
	- p\Delta \hat{\rm u}^r
	=
	 -{\rm div}\brkts{\hat{\mathbb D}^{0,r}(t,x)\nabla\ol{\rm v}^3}
	- {\rm div}\brkts{\tilde{\mathbb D}^{0,r}(t,x)\nabla\hat{\rm v}^3}
& \\\qquad\qquad\qquad\qquad
	+ z_r
	{\rm div}
		\brkts{
			\hat{\rm v}^r{\mathbb M}\nabla\ol{\rm v}^3
+			\tilde{\rm v}^r {\mathbb M}\nabla\hat{\rm v}^3
		}
	&  \quad\textrm{in }\Omega_T\,,
\\
\quad
-{\rm div}\brkts{\pmb{\epsilon}^0\nabla\hat{\rm v}^3}
	= p\brkts{\hat{\rm v}^1
	-\hat{\rm v}^2}
	&  \quad\textrm{in }\Omega_T\,.
\end{cases}
}{LAVM3}
The test function $\hat{\rm u}^r$, $r=1,2$, in \reff{LAVM3} together with Sobolev inqualities and 
$L^p$-interpolation estimates induce the following inequality,
\bsplitl{
\frac{p}{2}\frac{d}{dt} \brkts{\Ll{\hat{\rm u }^1}^2
		+ \Ll{\hat{\rm u }^2}^2}
	& + \brkts{p-\frac{\alpha}{2}}\brkts{\Ll{\nabla\hat{\rm u}^1}^2
		+\Ll{\nabla\hat{\rm u}^2}^2}
\\
	& \leq C(\alpha,\epsilon,\hat{\rm v}^r,\tilde{\rm v}^r)\brkts{\Ll{\hat{\rm v}^1}^2 +\Ll{\hat{\rm v}^2}^2}
	+\epsilon \brkts{\Ll{\nabla \hat{\rm v}^1}^2+\Ll{\nabla \hat{\rm v}^2}^2}\,,
}{Es1}
which is a consequence of the following estimates,
\bsplitl{
\av{\brkts{\hat{\mathbb{D}}^{0,r}\nabla\ol{\rm v}^3,\nabla\hat{\rm u}^r}}
	& \leq C\N{\ol{\rm v}^3}{H^2}\max_{1\leq k,l\leq N} \N{\hat{\rm D}^{0,r}_{kl}}{L^3}\Ll{\nabla\hat{\rm u}^r}
\\
	&\leq C\N{\ol{\rm v}^3}{H^2}\Ll{\hat{\rm v}^r}^{1/2}\Ll{\nabla\hat{\rm v}^r}^{1/2}\Ll{\nabla\hat{\rm u}^r}
\\
	&\leq C(\alpha,\epsilon)\N{\ol{\rm v}^3}{H^2}^2\Ll{\hat{\rm v}^r}^{2}
		+\epsilon\Ll{\nabla\hat{\rm v}^r}^{2}
		+ \frac{\alpha}{8}\Ll{\nabla\hat{\rm u}^r}^2	\,,
}{Es2}
and 
\bsplitl{
\av{\brkts{\tilde{\mathbb{D}}^{0,r}\nabla\hat{\rm v}^3,\nabla\hat{\rm u}^r}}
	& \leq C(\alpha) \N{\tilde{\rm v}^r}{H^1}^2\brkts{\Ll{\Delta\hat{\rm v}^3}\Ll{\nabla\hat{\rm v}^3}}
		+\frac{\alpha}{8}\Ll{\nabla\hat{\rm u}}^2
\\
	& \leq C(\alpha) \N{\tilde{\rm v}^r}{H^1}^2\Ll{\nabla\hat{\rm v}^3}\brkts{
			\Ll{\hat{\rm u}^1}^2 + \Ll{\hat{\rm v}^2}^2}
		+\frac{\alpha}{8}\Ll{\nabla\hat{\rm u}}^2
\,.
}{Es3}
It is immediately clear that the terms containing $\mathbb{M}$ can be controlled by the same bounds, since 
$\mathbb{M}$ is just a constant matrix. After integrating \reff{Es1} with respect to time we get,
\bsplitl{
& \frac{p}{2}\brkts{\N{\hat{\rm u }^1}{L^\infty(0,T;L^2(\Omega))}^2
		 + \N{\hat{\rm u }^2}{L^\infty(0,T;L^2(\Omega))}^2}
\\&
	+ \frac{\alpha}{2}\brkts{\N{\nabla\hat{\rm u}^1}{L^2(0,T;L^2(\Omega))}^2
		+\N{\nabla\hat{\rm u}^2}{L^2(0,T;L^2(\Omega))}^2}
\\&\qquad
	\leq C(\epsilon,\alpha,\hat{\rm v}^r,\tilde{\rm v}^r)T\brkts{\N{\hat{\rm v}^1}{L^\infty(0,T;L^2(\Omega))}^2 
		+\N{\hat{\rm v}^2}{L^\infty(0,T;L^2(\Omega))}^2}
\\&\qquad\quad
	+\epsilon\brkts{
		\N{\nabla\hat{\rm v}^1}{L^2(0,T;L^2(\Omega))}^2
		+\N{\nabla\hat{\rm v}^2}{L^2(0,T;L^2(\Omega))}^2
	}
		\,,
}{FiEs}
By choosing $T$ and $\epsilon$ such that $C(\alpha,\hat{\rm v}^r,\tilde{\rm v}^r)T \leq p/4$ 
and $\epsilon\leq \alpha/4$, we obtain a unique solution by Banach's fixed point 
theorem. The statement that ${\rm u}^3\in L^\infty(]0,T[;H^2(\Omega))$ is an immediate consequence 
of elliptic regularity theory.